\documentclass[11pt]{article}
 \usepackage[top=2.5cm,bottom=2.5cm,left=3cm,right=3cm]{geometry}
\usepackage[dvipsnames]{xcolor}
\usepackage{authblk}
\usepackage{graphics}
\usepackage{csquotes}
\usepackage{amsthm}
\usepackage{mathrsfs}
\usepackage{stackrel,graphicx}
\usepackage{enumitem}
\usepackage{booktabs}
\usepackage{comment}
\usepackage{braket}
\usepackage{subfigure}
\usepackage[small,justification=justified]{caption}
\usepackage{titlesec} 
\usepackage{multicol}
\usepackage{hyperref}
\usepackage{enumitem}
\usepackage{soul}

\usepackage[sorting=nty,style=nature,intitle=true,doi=true]{biblatex}
\hypersetup{urlcolor=blue, citecolor=red}
\usepackage[USenglish]{babel}

\usepackage[title]{appendix}
\usepackage{graphicx}
\usepackage{amsmath}
\usepackage{amssymb}
\usepackage{comment}
\usepackage{amsfonts}
\usepackage{color}
\usepackage[right]{lineno}
\linenumbers
\usepackage{cases}

\newcounter{lst}

\usepackage[normalem]{ulem}
\definecolor{sectioncolor}{rgb}
{0.7,0.3,0.3}
\definecolor{hillencolor}{rgb}{0.50, 0.03, 0.46}
\definecolor{stepiencolor}{rgb}{0.09, 0.45, 0.27}
\definecolor{contecolor}{rgb}{0.9, 0.3, 0.6}
\definecolor{ryancolor}{rgb}{0, 0.31, 0.78}
\definecolor{lightgray}{rgb}{0.1, 0.1, 0.1}
\newcommand{\RR}{\mathbb{R}}
\newcommand{\bnu}{{\bf \nu}}

\newcommand{\ep}{\varepsilon}

\definecolor{hintergrundfarbe}{rgb}{0.9,0.95,0.95}

\newtheorem{theorem}{Theorem}
\newtheorem{lemma}[theorem]{Lemma}
\newtheorem{corollary}[theorem]{Corollary}

\newtheorem{definition}{Definition}

\providecommand{\keywords}[1]{\textbf{\textit{Keywords --}} #1}

\newenvironment{sistem}
{\left\lbrace\begin{array}{@{}l@{}}}
{\end{array}\right.}

\hypersetup{
  pdftitle=Title,
  pdfpagemode=UseOutlines,
  citebordercolor=0 0 1,
  colorlinks=true,
  allcolors=sectioncolor,
  breaklinks=true,
  pdfauthor={Add authors here},
  pdfpagetransition=Dissolve,
  bookmarks=true
}

\usepackage{titlesec}

\titleformat{\section}
{\color{sectioncolor}\normalfont\Large\bfseries\sf}
{\color{sectioncolor}\sf\thesection}{1em}{}
\titleformat{\subsection}
{\color{sectioncolor}\normalfont\large\bfseries\sf}
{\color{sectioncolor}\sf\large\thesubsection}{1em}{}
\titleformat{\subsubsection}
{\color{sectioncolor}\normalfont\normalsize\bfseries\sf}
{\color{sectioncolor}\sf\large\thesubsubsection}{1em}{}

\begin{document}
\title{\textcolor{sectioncolor}{\sf Go-or-Grow Models in Biology:\\ a Monster on a Leash}}
\date{\today}
\author[a,$^\ddag$]{Ryan Thiessen}
\author[,b,$^\ddag$]{Martina Conte\footnote{Corresponding author: conte.martina.93@gmail.com\\ $^\ddag$These authors contributed equally to the work.}}
\author[c]{Tracy L. Stepien}
\author[a]{Thomas Hillen}
\affil[a]{\centerline{\small Department of Mathematical and Statistical Sciences, University of Alberta, Edmonton, Canada}}
\affil[b]{\centerline{\small Department of Mathematical, Physical, and Computer Sciences, University of Parma, Italy}}
\affil[c]{\centerline{\small Department of Mathematics, University of Florida, Gainesville, FL, USA}}

\setcounter{Maxaffil}{0}
\renewcommand\Affilfont{\itshape\small}
\maketitle

\begin{abstract}
    Go-or-grow approaches represent a specific class of mathematical models used to describe populations where individuals either migrate or reproduce, but not both simultaneously. These models have a wide range of applications in biology and medicine, chiefly among those the modeling of brain cancer spread. The analysis of go-or-grow models has inspired new mathematics, and it is the purpose of this review to highlight interesting and challenging mathematical properties of reaction--diffusion models of the go-or-grow type. 
    We provide a detailed review of biological and medical applications before focusing on key results concerning solution existence and uniqueness, pattern formation, critical domain size problems, and traveling waves. We present new general results related to the critical domain size and traveling wave problems, and we connect these findings to the existing literature. Moreover, we demonstrate the high level of instability inherent in go-or-grow models. We argue that there is currently no accurate numerical solver for these models, and emphasize that special care must be taken when dealing with the ``monster on a leash''.
\end{abstract}

\keywords{Go-or-grow dichotomy $|$  Glioma $|$ Traveling Waves $|$ Instabilities and Patterns $|$ Critical domain size} 

\paragraph{\indent \emph{\textbf{MSC codes}} --}{92B05, 35B36, 35M30}

\section{Introduction}
Various biological processes involve a dichotomy between proliferation and migration: individuals in a population either remain relatively stationary and reproduce, or they migrate with minimal reproduction. This phenomenon has been particularly emphasized in the study of glioma progression, where the ``go-or-grow" hypothesis was coined to describe the behavior of brain tumor cells \cite{hatzikirou2012}.

The seminal work by Murray et al.~\cite{murray2003-II} laid the groundwork for much of the existing mathematical modeling of glioma spread,  which is based on the classical Fisher--Kolmogorov--Petrovsky--Piskounov (FKPP) equation. This reaction--diffusion macroscopic mathematical framework has been termed as the proliferation--infiltration (PI) model, as it captures both the proliferative and invasive behaviors of glioma cells observed in experiments and clinical settings. The strength of the PI model lies in its simplicity, requiring only a small set of parameters, which facilitates its calibration to patient-specific data. Extensions of the PI model account for the influence of surrounding tissue, chemical signals, and/or vasculature on glioma cell density \cite{Swanson2003, jbabdi2005, saut2014, Hormuth2022, swanson2008mathematical, swanson2011quantifying, Conte2021PLOS,curtin2020mechanistic,curtin2020speed}. Advection bias in response to environmental cues has also been incorporated, either directly in the macroscopic model, based on mass/flux/momentum balance \cite{colombo2015, kim2009mathematical}, or through more detailed descriptions derived from lower-scale dynamics, as initially proposed in \cite{painter2013mathematical} and further developed in \cite{Conte2020, engwer2015, conte2021, engwer2016effective, engwer2015multiscale, swan2018patient, conte2023mathematical}. Many models have also been developed in (semi)discrete settings \cite{bottger2012investigation, hatzikirou2012, kim2013hybrid}. However, due to the complexity of glioma growth---spanning interactions between glioma cells, their microenvironment, immune cells, vasculature, and molecular signalling systems, all within the heterogeneous landscape of brain tissue---an exhaustive model remains elusive. As a result, models typically focus on well-defined subprocesses. Among these, the phenotypic heterogeneity of glioma cells plays a central role. While there are numerous reviews of mathematical modeling of gliomas \cite{hatzikirou-etal-2005, martirosyan-etal-2015, alfonso2017, jorgensen-etal-2023}, this paper specifically focuses on a family of reaction--diffusion models of the go-or-grow type. In these models, a distinct separation between proliferating and migrating cells leads to systems of two coupled reaction--diffusion equations.

We will focus our attention on a minimal go-or-grow model consisting of two coupled differential equations for the evolution of a moving subpopulation and a resting and proliferating subpopulation of cells. We assume linear diffusion for the moving population and neglect any small diffusivity in the resting population. Concerning the transition functions between the two subpopulations, we account for a general dependence on both subpopulation densities, defining this setting as the {\it general go-or-grow model}. As described above, there are several extensions of this framework that may include nonlinear diffusion of the moving compartment, general waiting time distributions, small diffusivity in the stationary compartment, or more complex population dynamics. Their analysis is rich, well-developed, and their analytical properties are vast \cite{murray2003-II,Fedotov2008,smoller,britton1986reaction,Kuznetsov2020}. 

\subsection{Paper outline}
We first provide in Section~\ref{sec:biomath_overview} a detailed overview of the use of go-or-grow models for describing glioma dynamics (Section \ref{sec:glioma}), reviewing the biological evidence of the go-or-grow dichotomy in glioma (Section \ref{sec:bioEvid}) and the mathematical literature of go-or-grow models formulated with discrete, continuous, deterministic and stochastic approaches (Section \ref{sec:previousmodels}). We also discuss applications of the go-or-grow dichotomy in ecology and other biological settings (Section \ref{sec:other}).

In Section \ref{sec:go-or-grow-models}, we first review the classical FKPP equation as a reference point for further modeling and analysis. We then introduce the main players of this review article, namely the {\it general go-or-grow model} and its simplifications, which we call the {\it total population go-or-grow model}, where the transition functions depend on the sum of the two subpopulation densities \cite{pham2012,stepien2018,tursynkozha2023traveling,Syga2024,bottger2015,falco2024traveling,hatzikirou2012}, and the {\it balanced go-or-grow model} and {\it constant rates go-or-grow model}, which fit into the context analyzed in \cite{Gerlee2012,gerlee2016traveling,Fedotov2008,falco2024traveling,bottger2015,pham2012,Syga2024}. We show that for fast transition rates, a FKPP-type model can be found as the leading order approximation.

After the models are defined, we focus on some mathematical properties. In Section \ref{sec:existence}, we review a method by Rothe~\cite{Rothe} to obtain a general local and global existence result. This does, however, not prevent strange behavior, as we show in Section \ref{sec:instabilities}. Here we employ methods that were developed by Marciniak-Czochra et al.~\cite{Anna1,Anna2}, to find that high-frequency instabilities are possible, and in fact, likely to occur. Section \ref{sec:criticaldomainsize} focuses on the classical ecological problem of a critical domain size, investigating what is the minimum size of a domain to support a go-or-grow population by reviewing previous results of Hadeler et al.~\cite{HadelerLewis2002,lewis1996biological}. Section \ref{sec:tw} then focuses on the existence of invasion waves and the estimation of the invasion speeds. We review previous results by Hadeler et al.~\cite{HadelerLewis2002}, Pham et al.~\cite{pham2012}, Stepien et al.~\cite{stepien2018}, and Falc\'{o} et al.~\cite{falco2024traveling}, and put them into an extended general framework. For this, we employ the general theory of invasion fronts of cooperative systems to generate a new result that applies in our more general framework. This section closes with a comparison of the wave speeds of the go-or-grow models to the FKPP wave speed.

This paper is rounded off with a discussion (Section~\ref{sec:discussion}) of the relevance of these models to both biology and mathematics and the mathematical challenges that are still open. The biological applications are manifold, and the mathematical problems are interesting and challenging. 

\section{Biological background and previous mathematical models} \label{sec:biomath_overview}

\subsection{Motivation: glioma dynamics}\label{sec:glioma}
Gliomas are the most frequent type of primary brain tumors. They originate from glial cells in the central nervous system and account for $78\%$ percent of malignant brain tumors, of which glioblastoma multiforme (GBM) is the most aggressive subtype. It is characterized by infiltrative spread and fast proliferation. These features make GBMs relentless and deadly brain cancers, claiming the lives of most patients within eighteen months of diagnosis, with only a small percentage surviving beyond five years \cite{ostrom2014,wick2018}.

A significant challenge stems from the highly invasive nature of gliomas, characterized by diffuse borders and recurrence of the disease around or within the treated area \cite{hou2006}. Glioma cells not only invade locally, navigating through the extracellular matrix (ECM), but also exploit structured pathways such as myelinated fibers, vasculature, and white matter tracts, enabling them to travel great distances \cite{cuddapah2014, Hormuth2022, saut2014, osswald2015, swan2018patient, engwer2015}.  This results in GBMs often presenting with a heterogeneous and anisotropic shape.

\subsubsection{Biological evidence of go-or-grow dichotomy in glioma}\label{sec:bioEvid}
Many experiments with cultures of glioma cells suggest mutual exclusion of migratory and proliferative behavior, a mechanism known as the go-or-grow dichotomy \cite{godlewski2010microrna,giese2003cost,xie2014targeting,berens1999those,giese1996,wang2012ephb2,farin-etal-2006}, schematized in Figure \ref{fig:experimentaldata}. Precisely, in \cite{giese2003cost,giese1996} experimental evidence suggests that there may be an inherent and inverse correlation between cell motility and proliferation within a malignant glioma cell population, i.e., highly migratory glioma cells have a lower proliferation rate compared to those cells with weak migratory ability that proliferate much more. In \cite{godlewski2010microrna}, a glioma-expressed microRNA that regulates the balance of proliferation and migration in response to metabolic stress is identified, revealing that, in addition to inhibiting glioma cell migration, this microRNA expression also promotes cell proliferation. The role of the receptor tyrosine kinase EphB2 in controlling the proliferation/migration dichotomy in GBMs and its pro-invasive and anti-proliferative effects is later shown in \cite{wang2012ephb2}. In \cite{hoek2008}, the go-or-grow behavior is also observed in melanoma cells, although other studies suggest that the behavior appears to be specific to tumors of glial origin \cite{garay-etal-2013}. 

However, the go-or-grow hypothesis does not apply to all tumor growth in general. For example, in \cite{ratliff2023individual,odde2023glioblastoma} the authors identified glioblastoma cells that exhibit both high proliferative and invasive capabilities simultaneously. In \cite{ayuso2017glioblastoma,scribner2017single}, the emergence of a go-and-grow phenotype under hypoxic conditions is observed. Similarly, in \cite{vittadello2020} it is shown for melanoma that differentiating between solely moving or proliferating cancer cells could not be determined. Hence, the go-or-grow mechanism is not a universal phenomenon in cancer, although it has been reported in many cases and most prominently in glioma.

\begin{figure}[t]
\centering
    \includegraphics[width=10cm]{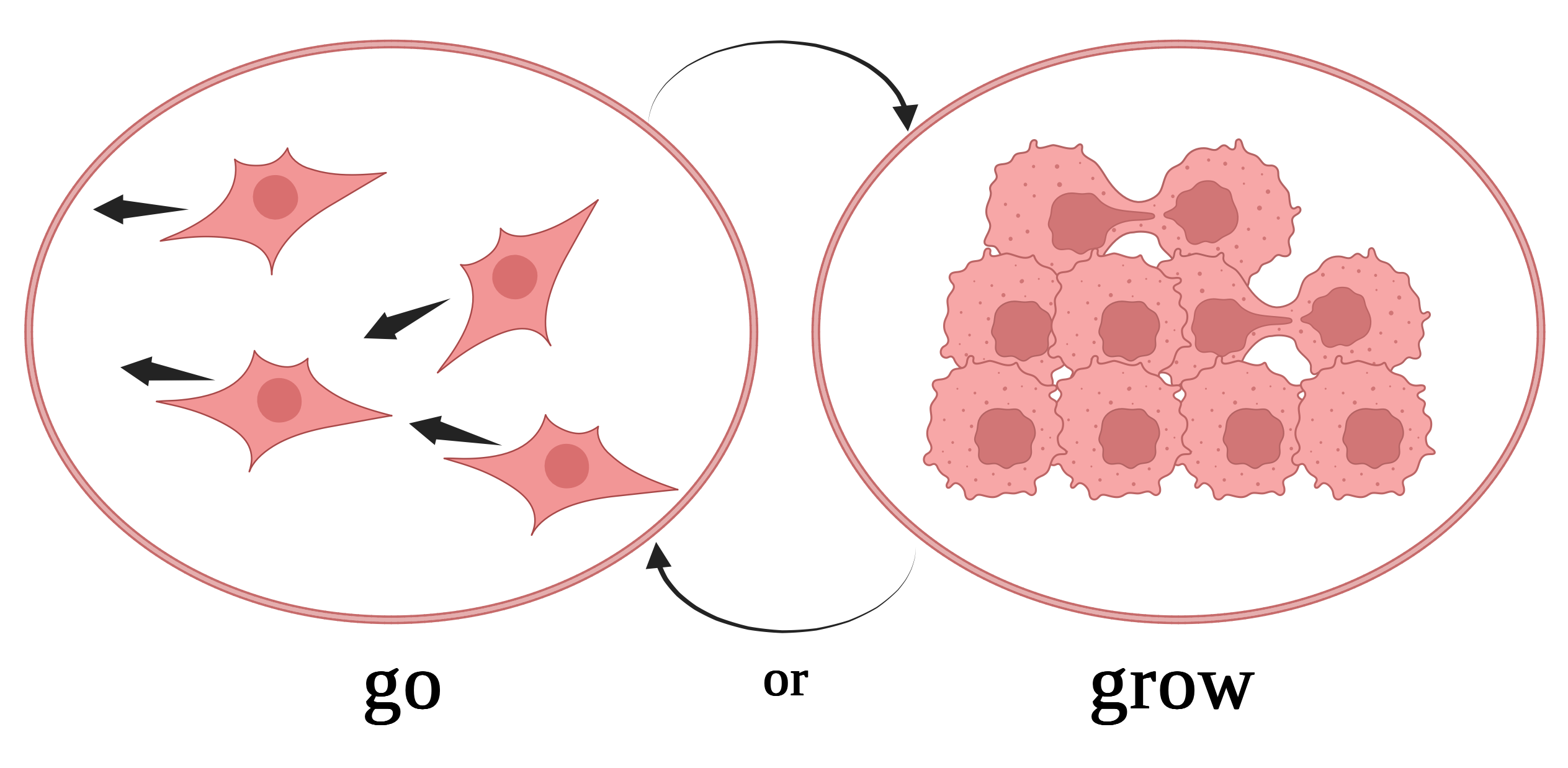}
    \caption{{\bf Sketch of the go-or grow dichotomy.} In this framework, cells can switch between a moving phenotype and a stationary and reproducing phenotype. Created in BioRender. Conte, M. (2024) \url{https://BioRender.com/i82o612}}
    \label{fig:experimentaldata}
\end{figure}

\subsection{Mathematical approaches to the go-or-grow dichotomy in glioma}\label{sec:previousmodels}
The large literature concerning mathematical models for the description of the go-or-growth mechanism in gliomas ranges from (semi)discrete to continuous approaches, from hybrid to stochastic methods. Moreover, within models of the same family, further classifications can be considered in relation to either the type of considered phenotypic switch functions (e.g.\ constant, density-dependent rates, or environmental-driven switching) or the characteristic of the operators modeling migration (e.g.\ linear or nonlinear diffusion, as well as additional tactic terms).

\subsubsection{Stochastic and discrete models}
One of the first publications investigating the go-or-grow mechanism is the model by Fedotov and Iomin \cite{Fedotov2008,fedotov2007migration} that uses a continuous time random walk (CTRW) to formulate the balance equations for the cancer cells of two phenotypes with random switching between cell proliferation and migration. The CTRW describes tumor cell transport, assuming an arbitrary waiting time distribution between jumps and general jump kernels, while proliferation is modeled by a logistic growth. The authors study the overall spreading rate, defined as the speed of the traveling wave solution of their system, and find its dependence on the statistical characteristics of the random switching process. From these first works, the literature on this topic started to expand.

A series of papers from Hatzikirou, Deutsch, and co-authors analyze the problem using lattice-gas cellular automaton (LGCA) models that account for unbiased cell random walk and cell proliferation influenced by the local oxygen concentration and according to a carrying capacity-limited mechanism. Phenotype changes either with constant probabilities \cite{bottger2012investigation,hatzikirou2012}, local cell density-dependent rates \cite{bottger2015,tektonidis2011identification,Syga2024}, or in response to environmental stimuli \cite{kim2013hybrid}. The width of the invasive zone and the speed of the invasive front are analyzed in relation to the phenotypic transition parameter. In particular, in \cite{bottger2015,Syga2024} the authors distinguish two scenarios depending on the parameter regions. In one case, called repulsive regime, the tumor cell population will inevitably grow. In the other case, called attractive regime, they identified conditions under which sufficiently small tumors die out. This is a phenomenon known as the  Allee effect and it is deeply studied in ecology. Mean field approximations of the LGCA models are also provided, under the assumption that the switch dynamics are much faster than cell proliferation and death. The authors derive macroscopic descriptions of FKPP type, showing that there is a trade-off between diffusion and proliferation, reflecting the inability of cells to migrate and proliferate simultaneously.

With a methodology similar to \cite{Fedotov2008}, in \cite{gerlee2016traveling,Gerlee2012} Gerlee and Nelander develop an individual-based (IB) stochastic model with constant switching rates between proliferating or migrating states. From it, they derive a system of macroscopic PDEs that establishes a connection between single cell characteristics and macroscopic behaviors. Precisely, their strategy is to formulate two coupled master equations for the occupation probabilities of the lattice sites and then approximate them by a set of PDEs. The derived macroscopic system features nonlinear diffusion for both cell populations, as well as additional transport for the moving cells, which makes the model different from the our minimal go-or-growth setting. Wave speed analysis is performed to show the influence of the parameters (switching, proliferation, motility, and apoptosis rate) on the wave speed. Moreover, under suitable simplified assumptions, a closed expression of the minimum wave speed is derived, as well as an approximate expression of the wavefront shape, using singular perturbation methods.

\subsubsection{Continuous models} \label{sec:continuous-models}

Inspired by the macroscopic counterparts of the above-described discrete approaches modeling the go-or-growth dichotomy in glioma cells, continuous formulations of such mechanism have been proposed. These approaches are either derived from multi-scale descriptions of cell dynamics or directly stated at the macroscopic level as reaction--diffusion models.

Concerning multi-scale models, the first setting is described in \cite{Chauviere2010}, where the authors build a model for two interacting cell populations with different migratory properties whose corresponding phenotype changes are induced by cell density effects. In particular, they investigate the phenotypic adaptation to dense and sparse environments, showing that when the cell migratory ability is reduced by a crowded environment, a diffusive instability may appear, leading to the formation of aggregates of cells of the same phenotype. Similar approaches to the same problem have been later proposed in a series of works accounting for the go-or-grow cell mechanism in response to constant or density-dependent rates \cite{engwer2015multiscale,zhigun2018strongly,HillenMicrotubes,alsisi2021non} as well as different environmental factor, such as oxygen and nutrient availability, acidity \cite{conte2021}, therapy effects \cite{hunt2017multiscale}, or cell membrane receptors' concentration \cite{stinner2016global}. These approaches feature non-linear diffusion and transport terms for the moving cells. Moreover, in some of these works, results concerning the global existence of weak solutions for the derived macroscopic setting are proved \cite{stinner2016global,zhigun2018strongly}. In this category, we also have to mention a recent work from Hillen and co-authors \cite{HillenMicrotubes}, who propose a mathematical model for the microtube-induced glioma invasion in which, under suitable scaling arguments, a go-or-grow mechanism could be identified, {\it a posteriori}, justifying the earlier approaches. The idea of microtube-mediated cell invasion comes from some recent experiments \cite{osswald2015} showing how glioma cells send out thin cell protrusions - the tumor microtubes - that form a network invading the healthy tissue, making connections between cells and enabling long-range cell-to-cell signaling.

Considering models directly stated at the macroscopic level in terms of reaction--diffusion systems of equations for the two cell sub-populations, a first approach can be found in \cite{pham2012}. The described setting models a density-dependent go-or-grow mechanism, characterizing the development of spatio-temporal instabilities and traveling wave solutions. In the same context of density-dependent phenotypic changes, we must also account for the investigations collected in \cite{stepien2018,tursynkozha2023traveling,falco2024traveling}. Precisely, in \cite{stepien2018} the authors analyze the speed and shape of traveling wave solutions of a go-or-grow model for some specific expression of the phenotypic transition functions. They derive a closed form for the minimum wave speed and, adapting the method of Canosa \cite{canosa1973nonlinear}, obtain an approximation of the traveling wave solution. This study has been then extended in \cite{tursynkozha2023traveling} for slightly more general phenotypic transition functions. Recently, in \cite{falco2024traveling}, a formal analysis of the traveling wave solutions has been presented in comparison with the classical results of the well-known FKPP equation. However, this has been obtained with simplified assumptions on the transition parameters, comparable with considering constant rate functions. Besides these studies, which relate to a minimal go-or-grow model similar to \eqref{general-model}, it is worthy to mention some other works developed within the same biological context that, however, either account for small diffusion terms in the equation for the proliferative population \cite{burger2020reaction,pardo2016nonlinear} or consider phenotypic transition rates depending on external signals \cite{alfonso2016one,saut2014,martinez2012hypoxic,stein2007mathematical,stein2018mathematical,crossley2024phenotypic,ayuso2017glioblastoma}. Models that incorporate the microenvironment include \cite{crossley2024phenotypic} where the authors analyze a heterogeneous population comprising cells that can either move and degrade ECM or proliferate, including both nonlinear diffusion and an additional tactic term for the former, as well as the evolution of the ECM. They explore the impact of different switching mechanisms on the speed of the populations, including constant, ECM-dependent, and density-dependent switching rates.

\subsection{Other applications} \label{sec:other}
Go-or-grow models have also been used in ecological contexts for the description of immobile and dispersing populations and two-species competition settings. Studies of the ``drift paradox'', where aquatic insects in streams with downstream drift are able to remain upstream, have included two-phase models where an insect is either mobile in the stream (moving with diffusion and advection) or stationary on the benthos (and reproducing), and insects transition between each phase \cite{pachepsky-etal-2005,lutchser-etal-2006,lutchser-etal-2007,lutscher-seo-2011,huang-etal-2016,liang-etal-2017,wang-shi-2020,yan-etal-2022}. The benthic-drift model has also been extended to study the movement of zebra mussels in rivers \cite{huang-etal-2017,jin-zhao-2021}. Some models of predator-prey systems or spreading of epidemics take on the form of two-equation reaction--diffusion systems with one population as the prey or the susceptible population and the other as the predator or the infected population, respectively \cite{capasso1981convergence,dunbar1983traveling,dunbar1984traveling,dunbar1986traveling,huang2003existence,li2008traveling,lin2011traveling,hsu2012existence,huang2012traveling,huang2016geometric,zhang2016minimal,zhao-zhou-2016,lam2018traveling}. Here, the populations only transition in one direction (from energy conversion from prey to predator or transfer from susceptible to infected population) and both may be mobile.

Other examples of go-or-grow mechanisms in biological modeling studies include movement of forest boundaries \cite{kuznetsov-etal-1994,chuan-etal-2006}, amoebae \cite{cochet-escartin-etal-2021}, genetically engineered microbes \cite{lewis-GEM-etal-1996}, cells during wound healing \cite{landman-etal-2007}, genetic traits \cite{venegas-ortiz-etal-2014}, mRNA cargo \cite{ciocanel-etal-2018}, and bacteria \cite{patra-klumpp-2021}.

Further afield, studies of autocatalytic chemical reactions often include reaction--diffusion equation systems with one-directional transitions and equal or differing diffusion constants for the reactant and autocatalyst  \cite{billingham-needham-1991,merkin-sadiq-1996,merkin-etal-1996,sadiq-merkin-1996,sadiq-merkin-1996,finlayson-merkin-1997,xin-2000}. The reaction terms in the autocatalytic models inspired the form of the Klausmeier--Gray--Scott model, which examines stripe and pattern formation of vegetation depending on water movement (\cite{vanderstelt-etal-2013,sewalt-doelman-2017,chen-wang-2024}).

\section{The go-or-grow model} \label{sec:go-or-grow-models} 

In this section, we first introduce the classical FKPP equation \cite{murray2002-I,Levin,britton1986reaction}, for which there is a vast developed literature and because it forms the basis for analysis of go-or-grow models due to their similar structure. Various go-or-grow models are subsequently stated, and the connection between the FKPP equation and go-or-grow models is demonstrated in the last subsection via the fast transition rate scaling.

\subsection{FKPP equation} 
Letting $m(t,x)$ denote the population size at time  $t\geq 0$ and spatial position $x\in \Omega\subset\mathbb{R}^k$, $k\geq 1$,  then the FKPP reaction--diffusion equation reads 
\begin{equation}\label{FKPP}
m_t = d \Delta m + f(m), 
\end{equation}
where the subscript $t$ denotes the partial time derivative and $\Delta$ denotes the spatial Laplacian on $\Omega$. Typically, $f(m)$ is of \ul{\it logistic type}, i.e., it satisfies the following properties
\begin{equation}\label{logistic}
f(0) = 0, \quad f(K) = 0,\quad f'(0)>0,\quad f'(K)<0, \,\,\,\mbox{and}\,\, f(m)>0 \mbox{ for } 0<m<K, 
\end{equation}
where $K>0$ denotes the carrying capacity. In addition, we will often consider the \ul{\it subtangential condition}, which is 
\begin{equation}\label{st}
f'(0) m \geq f(m),\qquad \text{for }\,\, 0\le m\le K. 
\end{equation}
In the following, we will use the notation $
{f(m)=g(m)m}$, where $g(m)$ denotes the per-capita growth rate of the population $m(t,x)$.

\subsection{Go-or-grow model types}
Go-or-grow models can be seen as a generalization of the FKPP equation. We first define the general continuous go-or-grow model that represents the focus of this review, and then several special variations that are used to address some of the mathematical challenges described in the subsequent sections. 

\paragraph{General go-or-grow model.} Let $u(t,x)$ denote the density of migrating cells and $v(t,x)$ the density of stationary and proliferating cells. The {\it \ul{general go-or-grow model}} is the system of equations
\begin{equation} \label{general-model}
\begin{sistem}
u_t = d\Delta u -\mu u - \alpha(u,v) u +\beta(u,v) v, \\[0.2cm]
v_t = g(u+v) v +\alpha(u,v) u -\beta(u,v) v ,
\end{sistem}
\end{equation}
where $d>0$ is a constant diffusion coefficient, $\mu\geq 0$ is a natural death rate for the moving compartment, $\alpha(u,v)$ denotes the transition rate function from the moving cell compartment to the stationary one, and $\beta(u,v)$ denotes the transition rate function from the stationary cell compartment to the moving one. The function $g(u+v)v$ accounts for the proliferation of the cell population, and, in the literature, it often takes a logistic form.

\paragraph{Total population go-or-grow model.} Starting from the general model \eqref{general-model}, a first variation of it is obtained by assuming that the transition rates $\alpha$ and $\beta$ depend on the sum of the two populations, i.e., $\alpha=\alpha(n)$ and $\beta=\beta(n)$, with 
\begin{equation}\label{n_def}
    n(t,x) = u(t,x) + v(t,x)\,.
\end{equation}
This version of the go-or-grow model has been used in many papers (e.g. \cite{pham2012,falco2024traveling,Syga2024,tektonidis2011identification,tursynkozha2023traveling}). We refer to it as the {\it  \ul{total population go-or-grow model}}, whose equations are given by
\begin{equation}\label{totalpop-model}
\begin{sistem}
    u_t = d\Delta u -\mu u - \alpha(n) u +\beta(n) v, \\[0.2cm]
    v_t = g(n) v +\alpha(n) u -\beta(n) v.
\end{sistem}
\end{equation}

\paragraph{Balanced go-or-grow model.} If we consider that the transition probabilities from moving to stationary compartment and from stationary to moving compartment are balanced in the total population model~\eqref{totalpop-model}, i.e., $\alpha(n)+\beta(n)=1$, we obtain the {\it \ul{balanced go-or-grow model}}. Used in several works (e.g. \cite{pham2012,bottger2015,Syga2024}), its equations read
\begin{equation} \label{balanced-model}
\begin{sistem}
    u_t = d\Delta u -\mu u + \gamma[\Gamma(n) v - (1-\Gamma(n)) u],\\[0.2cm]
    v_t = g(n) v -\gamma[\Gamma(n) v - (1-\Gamma(n)) u],
\end{sistem}
\end{equation}
where $\gamma>0$ denotes the rate at which a switch happens and $\Gamma(n)$ and $(1-\Gamma(n))$ are the probabilities that the switch is from moving to stationary compartment or vice versa, respectively.

\paragraph{Constant rates go-or-grow model.} A further simplification of the previous models can be considered by assuming that the transition rates are constant, such as in \cite{lewis1996biological,gerlee2016traveling,Fedotov2008}. We refer to this model as the {\it \ul{constant rates go-or-grow model}}, which reads
\begin{equation} \label{constantrates-model}
\begin{sistem}
    u_t = d\Delta u -\mu u - \alpha u +\beta v, \\[0.2cm]
    v_t = g(u+v) v +\alpha u -\beta v.  
\end{sistem}
\end{equation}
The constant rates model~\eqref{constantrates-model} may be seen as a special case of the balanced model~\eqref{balanced-model} by setting 
\begin{equation}\label{cond_models}
    \gamma=\alpha+\beta \quad \mbox{and} \quad \Gamma=\frac{\beta}{\alpha+\beta} .
\end{equation}
Analogously, the balanced model~\eqref{balanced-model} may be seen a special case of the total population model~\eqref{totalpop-model} using the same assumption \eqref{cond_models} and considering the rate $\gamma$ depending on the total population.\\
\indent For later use, it is useful to write the above models in vector form as
\begin{equation}\label{vector}
{\bf w}_t = D \Delta {\bf w} + {\bf F}({\bf w}), 
\end{equation}
where ${\bf w} = (u,v)^T$, $D = \mathrm{diag}(d,0)$, and 
\begin{align}\label{F_vec}
    {\bf F}({\bf w}) = \begin{pmatrix}
        -\mu u -\alpha(u,v)u + \beta(u,v)v \\[0.3cm]
        g(u+v)v + \alpha(u,v)u -\beta(u,v)v 
    \end{pmatrix}\,.
\end{align}
In this formulation, (\ref{vector}) clearly resembles (\ref{FKPP}). 

\subsection{Fast transition rate scaling}\label{sec:fast_trans_rate}
We begin by reviewing a well-known scaling technique used to connect the mathematical properties of the go-or-grow type of models with the classical FKPP equation \eqref{FKPP}.
The \ul{\emph{fast transition rate scaling}} is based on the natural assumption that the transitions between stationary and moving states are fast as compared to proliferation and diffusion. Letting $\ep > 0$ be a small parameter, we consider the general go-or-grow model \eqref{general-model} with fast transition rates $\alpha \sim \alpha/\ep$ and $\beta \sim \beta/\ep$, resulting in
\begin{equation}\label{Scaling_FTR}
\begin{sistem}
    u_t = d\Delta u -\mu u - \dfrac{1}{\ep}\left(\alpha(u,v) u -\beta(u,v) v\right), \\[0.4cm]
    v_t = g(u+v) v +\dfrac{1}{\ep}\left(\alpha(u,v) u -\beta(u,v) v \right).
\end{sistem}
\end{equation}
\noindent We introduce the regular perturbation expansion of $u(t,x)$ and $v(t,x)$,
\begin{align*}
    u(x,t;\ep) &= u_0(t,x) + \sum^{\infty}_{n = 1}u_n(t,x) \ep^n, \\
    v(x,t;\ep) &= v_0(t,x) + \sum^{\infty}_{n = 1}v_n(t,x) \ep^n,
\end{align*}
and we assume that $\alpha$, $\beta$, and $g$ are of order O(1) and sufficiently smooth to allow for a Taylor expansion about $(u_0,v_0)$. Further assuming that each order of $\ep$-terms vanishes independently, the zeroth order of $\ep$ reads
\begin{align} \label{perturbexpand_order0}
     -\alpha(u_0,v_0)u_0 + \beta(u_0,v_0)v_0 = 0\,, 
\end{align}
while the first order of $\ep$ gives
\begin{subequations}
\begin{align}
    (u_0)_t  &= d\Delta u_0 -\mu u_0 - \alpha(u_0,v_0)u_1 +\beta(u_0,v_0)v_1 \nonumber \\
    & \qquad - (u_1,v_1)\cdot\nabla_{uv}\alpha(u_0,v_0)u_0
    + (u_1,v_1)\cdot\nabla_{uv}\beta(u_0,v_0)v_0,  \label{perturbexpand_order1_a}\\[0.3cm]
    (v_0)_t  &= g(u_0+v_0)v_0 + \alpha(u_0,v_0)u_1 - \beta(u_0,v_0)v_1 \nonumber \\
    & \qquad + (u_1,v_1)\cdot\nabla_{uv}\alpha(u_0,v_0)u_0
    - (u_1,v_1)\cdot\nabla_{uv}\beta(u_0,v_0)v_0 ,  \label{perturbexpand_order1_b}
\end{align}
\end{subequations}
where $\nabla_{uv}$ is used to distinguish the derivatives with respect to $u$ and $v$ from the spatial derivatives.
Summing \eqref{perturbexpand_order1_a} and \eqref{perturbexpand_order1_b}, and then together with \eqref{perturbexpand_order0}, we get the system
\begin{equation}\label{gen-fast-transition1} 
\begin{sistem}
    (u_0+v_0)_t = (d \Delta - \mu) u_0  + g(u_0+v_0)v_0,  \\[0.4cm]
    -\alpha(u_0,v_0)u_0 + \beta(u_0,v_0)v_0 = 0. 
\end{sistem}
\end{equation}

\noindent The second equation in \eqref{gen-fast-transition1} is an implicit equation for $u_0$ and  $v_0$. We can use the implicit function theorem to find $v_0$ as a function of $u_0$ if
\[
\dfrac{\partial \beta (u_0,v_0)}{\partial v_0}v_0+\beta (u_0,v_0)-\dfrac{\partial \alpha (u_0,v_0)}{\partial v_0} u_0\ne 0 .
\]
For particular forms of the transition rate functions $\alpha$ and $\beta$, it is possible to explicitly compute the inverse function. 

The first example of using the fast-transition rate scaling in go-or-grow models appears in \cite{HadelerLewis2002} for the constant rates model~\eqref{constantrates-model}. Under this model, it is possible to simplify system \eqref{gen-fast-transition1} to a single equation for the total population $n$, obtaining a convex combination of diffusion and growth as 
\begin{align}\label{FKPP_FTS}
    n_t  = \frac{\beta}{\alpha +\beta}d\Delta n + \frac{\alpha}{\alpha + \beta}g(n)- \mu n .
\end{align}
Here, the ratios $\beta/(\alpha + \beta)$ and $\alpha/(\alpha + \beta)$ represent the average fractions of the total cell population that belongs to the moving and stationary phase, respectively.

In the context of the total population go-or-grow model~\eqref{totalpop-model}, an example of fast-transition rate scaling is done in \cite{falco2024traveling} with $\mu=0$ and $d=1$. The following single governing equation in terms of the total population $n$ is obtained:
\begin{equation}
    n_t = \nabla \cdot \left(D(n)\nabla n\right) + \frac{\alpha(n)}{\alpha(n)+\beta(n)}n(1-n),
\end{equation}
where the density-dependent diffusion takes the form
\begin{equation}
    D(n) = \dfrac{\alpha(n)\beta'(n)- \alpha'(n)\beta(n)}{(\beta(n)+ \alpha(n))^2} \dfrac{\beta(n)}{\beta(n)+ \alpha(n)} .
\end{equation}
The natural assumptions that $\alpha(n)$ is a non-increasing function of the total cell density and $\beta(n)$ is a non-decreasing function guarantee that $D(n) \geq 0$ for any $n\geq 0$. If $\alpha$ and $\beta$ are constant, we return to equation \eqref{FKPP_FTS}, which is of FKPP type \eqref{FKPP}. Moreover, under the assumption that cells remain proliferative at low densities, i.e., $\beta(0)=0$, a diffusion coefficient is obtained which resembles typical non-linear equations with degenerate diffusion. 

Finally, in \cite{scribner2017single}, the authors compare a go-or-grow type model to a single-population model, showing how the single-population model replicate key features of glioblastoma progression. While their go-or-grow model includes switching functions based on factors other than cell density, it is important to note that the parameters used for comparison in the go-or-grow model (see Table 2 in \cite{scribner2017single}) are an order of magnitude larger than the other rates. This discrepancy satisfies the assumption of the fast transition rate approximation, which helps explain why the simulation results from both models are similar.

\section{A general existence result} \label{sec:existence}
All the models introduced in the previous section are based on the coupling of an ordinary differential equation with a reaction--diffusion equation. It is possible to prove a general existence result by using the existence theory for mixed type equations presented by Rothe in \cite{Rothe}. Here, we recall the main assumptions and definition.

Let us consider system \eqref{general-model} in the form 
\begin{equation}\label{generalRothe}
\begin{sistem}
    u_t = d\Delta u -d u +(d-\mu) u - \alpha(u,v) u +\beta(u,v) v, \\[0.2cm]
    v_t = g(u+v) v +\alpha(u,v) u -\beta(u,v) v, 
\end{sistem}
\end{equation}
where we explicitly write the terms $-du + du$ so that $d(\Delta-1)$ can be identified as a generator that has no eigenvalue equal to 0 in typical Neumann and periodic boundary conditions.
In the case of Dirichlet boundary conditions, we do not need to use this trick and instead consider $d\Delta$ as the generator.
\begin{definition}
    A \ul{\emph{mild solution}} $(u,v)$ of \eqref{generalRothe}  satisfies
    \begin{eqnarray*} (u,v)&\in& L^\infty([0,T)\times\Omega)^2 \\
    u(\cdot,t) &=& S(t) u_0 +\int_0^t S(t-s)\Bigl((d-\mu)u(s) -\alpha(u(s), v(s)) u(s) + \beta(u(s), v(s)) v(s) \Bigr) ds\\
    v(\cdot,t) &=& v_0 +\int_0^t g(u(s)+v(s))v(s)  +\alpha(u(s), v(s)) u(s) - \beta(u(s), v(s)) v(s)\; ds         
    \end{eqnarray*}
where $u(s) = u(\cdot,s)$ and $v(s) = v(\cdot,s)$, while $S(t)$ denotes the semigroup generated by $d(\Delta-1)$ on $L^\infty(\Omega)$ with the boundary conditions \eqref{BOU} defined below.
\end{definition}

Let us introduce the following set of assumptions.
{\allowdisplaybreaks
\begin{subequations}
\begin{align}
\begin{split}
\bullet&\,\,\Omega\subset\RR^n\,\text{is a smooth bounded domain with H\"older continuous boundary, i.e.,}\\
&{\hspace{5.3cm}\partial\Omega\in C^{2+\alpha}\,,\,\,\,\, \alpha\in(0,1)}.
\end{split}\label{DOM}\\[0.2cm]
\bullet&\, \,\text{The initial conditions satisfy }\, \bar{u}_0, \bar{v}_0\in L^\infty(\Omega)\,.\label{INI}\\[0.2cm]
\begin{split}
\bullet&\,\, \text{There are differentiable functions }\, a(x)\geq 0\,,\,b(x)\geq 0\,,\,a(x)\, \text{ and } b(x)\, \text{not simultaneously}\\ 
&\text{ equal to 0, such that for all } x\in \partial \Omega,\\
&{\hspace{4.5cm} b(x) u(t,x) +a(x) \frac{\partial}{\partial n} u(t,x) =0,}\\[0.1cm]
&\text{ where $n$ denotes the outer normal vector on } \partial\Omega.
\end{split}\label{BOU}\\
\begin{split}
\bullet&\,\, \text{For each bounded set }\, B\in \RR^2\,, \text{ there exists a constant }\, C_B\,\text{ such that for all }\\
& {(u,v), (w,y)\in B}, \\
&\hspace{3.5cm} |g(u+v)+\alpha(u, v) +\beta(u,v) | \leq C_B\,,\\
&\hspace{3.5cm} |g(u)- g(w)|\leq C_B|u-w|\,,\\
&\hspace{3.5cm} |\alpha(u,v)-\alpha(w,y)|\leq C_B (|u-w| + |v-y|)\,,\\
&\hspace{3.5cm} |\beta(u,v) - \beta(w,y)| \leq  C_B (|u-w| + |v-y|)\,.     
\end{split}\label{LIP}\\[0.2cm]
\begin{split}
\bullet&\,\text{ The functions }\, g,\alpha,\beta \in C^2;\,\text{ moreover, }\, \bar{u}_0\in C^{2+\alpha}\,,\,\bar{v}_0 \in C^{\alpha}, \text{ and } \bar{u}_0\, \text{ satisfies the}\\
&\text{ boundary conditions \eqref{BOU}.} 
\end{split}\label{STR}
\end{align}
\end{subequations}}

\begin{theorem}\label{Rothe}[Theorem 1 in \cite{Rothe} (pages 111-112)]
Assume \eqref{DOM}--\eqref{LIP}, then \eqref{general-model} has a unique mild solution which is bounded in $L^\infty$. If the maximal time of existence $T_{\tiny\mbox{max}} <\infty$, then 
\[ \lim_{t\to T_{\tiny\mbox{max}}} \| (u(t,x), v(t,x))\|_\infty = \infty. \]
Moreover, if \eqref{STR} holds, then, for all $T<T_{\tiny\mbox{max}}$, the solutions are classical with 
\[ u\in C^{2+\alpha,1+\alpha/2}(\bar \Omega\times [0,T)), \qquad v\in C^{\alpha, 1+\alpha/2} (\bar\Omega\times [0,T))\,. \] 
\end{theorem}

\noindent Theorem \ref{Rothe} applies to all our variants of the go-or-grow models \eqref{general-model}--\eqref{constantrates-model} with appropriate adjustments in the assumptions on the parameters. If a small diffusion term ($\epsilon\Delta v$) is considered in the equation for the proliferative population, e.g., as done in \cite{burger2020reaction}, then standard semigroup methods can be applied to find solutions to the system (see \cite{Robinson,HillenFA}).

\section{Instabilities and pattern formation}\label{sec:instabilities}

It is time to meet the monster on a leash. Marciniak-Czochra and collaborators have worked intensively on systems of reaction--diffusion equations coupled to ODEs \cite{Anna1,Anna2}. In the study of Turing-like pattern formation, they made an interesting discovery, which is that all smooth patterns are unstable, and non-constant steady states need to have jumps. They also found that high frequencies are unstable in the linearization. This has far-reaching consequences, also for our study here, as we will outline in this section. 

\subsection{The general case} 
The ability of reaction--diffusion systems to form spatial patterns is a well researched area \cite{turing1952chemical,murray2003-II,britton1986reaction,Levin}. In a typical activator-inhibitor situation, for example, we consider a system of two reaction--diffusion equations, where one species is the activator and the other the inhibitor. Turing patterns arise for a short range activator and a long range inhibitor \cite{turing1952chemical,murray2003-II,britton1986reaction,Levin}. The extreme situation would be a fully local activator, i.e., an activator without a diffusion term in the equations, leading to an ODE coupled to the PDE of the inhibitor. Hence, it is natural to expect that ODE-PDE systems can generate diffusion driven instabilities \cite{Anna1,Anna2}, and we review some of the results here. 

For a coupled system of the form 
\begin{equation}\label{MCmodel}
\begin{sistem}
    u_t = d_u \Delta u + k(u,v) ,\\[0.2cm]
    v_t = h(u,v) ,
\end{sistem}
\end{equation}
on a smooth domain with standard boundary conditions,  a necessary and sufficient condition for diffusion driven instability at a steady state $(\bar u, \bar v)$ is the \ul{\it autocatalysis condition}
\begin{equation}\label{autocatalysis}
\frac{\partial h}{\partial v } (\bar u, \bar v) >0. 
\end{equation}
A proof is given in Theorem 4 in \cite{Anna1}, and we sketch some of the important steps here. If we abbreviate the partial derivatives at steady state as 
\[ k_u = k_u(\bar u, \bar v), \qquad k_v=k_v(\bar u, \bar v), \qquad h_u=h_u(\bar u, \bar v), \qquad h_v = h_v(\bar u, \bar v),\]
then the Jacobian of the kinetic part of (\ref{MCmodel}) evaluated at $(\bar u, \bar v)$, together with its trace and determinant is 
\[ A := \left(\begin{matrix}k_u & k_v\\ h_u & h_v\end{matrix} \right), \qquad\mathrm{tr} A = k_u + h_v,\qquad \det A = k_u h_v- k_v h_u. \]
To study diffusion-driven instability, we assume that $(\bar u, \bar v)$ is stable for the kinetic part, i.e. 
\begin{equation}\label{linstable}
\mathrm{tr} A<0 \quad \mbox{ and } \quad \det A > 0.
\end{equation}
We linearize the full system \eqref{MCmodel} at the homogeneous steady state $(\bar u, \bar v)$, and we use the Fourier transform with dual variable $\omega$. The dual variable $\omega$ plays the role of the mode. For example, if we restrict on a given interval $[0,L]$, then we have only discrete modes $\omega_n = \frac{n\pi}{L}$ for $n\in \mathbb{N}$. The use of an abstract variable $\omega$ here allows us to make the following calculations as if $\omega$ were a continuous variable, while later, we can restrict to the discrete values $\omega_n$. The Jacobian of \eqref{MCmodel} at $(\bar u, \bar v)$ becomes 
\[ J = \left(\begin{matrix} -d_u \omega^2 + k_u & k_v \\ h_u & h_v\end{matrix}\right),\]
which has trace and determinant 
\[ \mathrm{tr} J = -d_u \omega^2 +\mathrm{tr} A <0,\qquad \det J = -d_u \omega^2 h_v + \det A.\]
Hence we observe diffusion driven instability if $\det J<0$, i.e.,
\begin{equation}\label{omegacond}
    \omega^2 > \frac{\det A}{d_u h_v }. 
\end{equation}
Under this assumption, we find the corresponding eigenvalues of $J$ as 
\[ \lambda_{1,2} =\frac{\mathrm{tr} A}{2} \pm \frac{1}{2} \sqrt{\mathrm{tr} J^2 - 4 \det J},\]
where $\lambda_1>0$ and $\lambda_2<0$. It is interesting to look at the unstable eigenvalue $\lambda_1$ as a function of the mode $\omega$,
\begin{equation}\label{lambdaomega}
\lambda_1(\omega) =\frac{1}{2}(-d_u \omega^2 h_v + \mathrm{tr} A) +\frac{1}{2} \sqrt{ (-d_u \omega^2 h_v +\mathrm{tr} A) ^2 - 4(-d_u \omega^2 h_v + \det A )} .
\end{equation}
Using elementary calculations, we can show (skipping the details) that for 
\[ \omega=\bar \omega =\sqrt \frac{\det A}{d_u h_v } , \]
we have $\lambda_1(\bar \omega)=0$. Moreover, $\lambda_1(\omega)$ is strictly increasing for $\omega\geq \bar \omega$ and it goes into saturation 
\[ \lim_{\omega\to \infty} \lambda_1(\omega) = 2( 2-\mathrm{tr} A) >0 .\]
Hence, $\lambda_1(\omega)$ has the qualitative form shown in Figure~\ref{desmos}, which illustrates that higher modes are more unstable than smaller modes. This is the main reason why smooth non-constant steady states are unstable \cite{Anna1,Anna2}. Moreover, this observation also has direct consequences for numerical computation. Note that different choices of parameters than those chosen in Figure~\ref{desmos} give qualitatively the same form. 

\begin{figure}[!h]
    \centering
    \includegraphics[width=.5\textwidth]{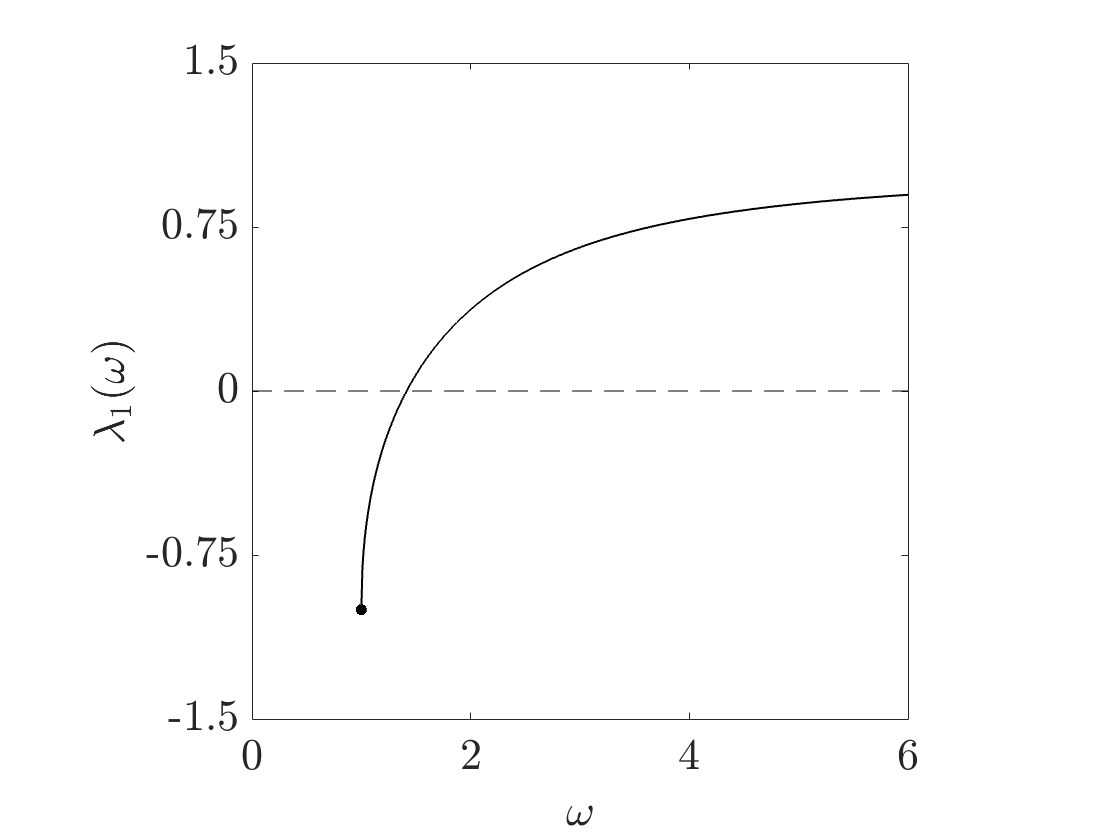}
    \caption{{\bf Qualitative trend of $\lambda_1(\omega)$.} Typical form of the unstable eigenvalue $\lambda_1(\omega)$ computed from \eqref{lambdaomega} as function of the mode $\omega$. Here, the parameters are set to $d_u=1$, $k_u=-2$, $k_v=-2$, $h_v=1$, and $h_u=2$.}
    \label{desmos}
\end{figure}
Now, we consider a discretization of our domain with a step size $\Delta x$. Then the spatial frequency is $\Delta x^{-1}$. According to what we just observed for $\lambda_1(\omega)$, each mode $\omega>\Delta x^{-1}$ is unstable. For example, using homogeneous Neumann boundary conditions on an interval $[0,L]$, then $\omega_k =k\pi/L$ and each mode $k$ with 
$\omega_k>\Delta x^{-1}$ is unstable. Hence, in the Neumann case, and also in other cases, small errors on the scale of the numerical discretization are amplified, making the numerics unstable. Reducing the step size does not help, as then higher modes are still unstable. This appears to be a principal challenge of the ODE-PDE coupled systems, and in such a situation, an accurate numerical solver is inconceivable. 
In Figure \ref{Turing1}, we show a simulation of the balanced model \eqref{balanced-model} on a rectangular domain of size $[0, 50]\times[0,50]$ for the transition and growth functions 
\begin{equation}\label{Func_Pham}
    \Gamma(n)=\dfrac{1}{2}\left(1+\tanh[\nu(n^*-n)]\right),\qquad g(n)=r\left(1-\dfrac{n}{K}\right) .
\end{equation} 
In this case, the above conditions \eqref{autocatalysis}, \eqref{linstable}, and \eqref{omegacond} for diffusion-driven instability are satisfied. We simulate the system \eqref{balanced-model} by using the method of lines and discretizing in space using a hybrid finite volume finite difference scheme, adapted from \cite{thiessen2021anisotropic}.
The solution converges to a smooth solution for $u$ and a spatial pattern for $v$ that is on the size of the discretization. 

\begin{figure}[!h]
\centering
    \includegraphics[width=\linewidth]{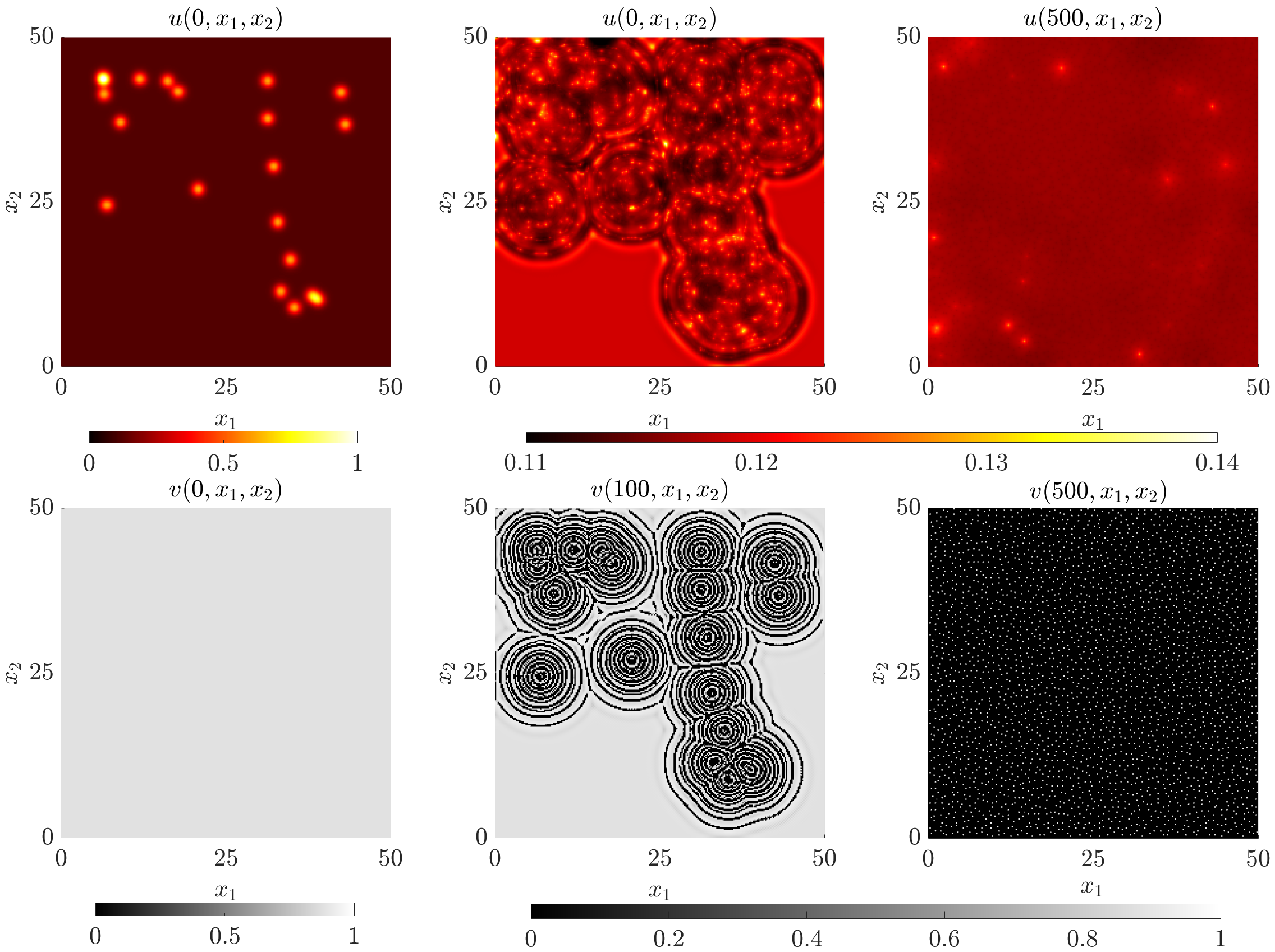}
    \caption{{ \bf  Example of 2D instabilities emerging from model \eqref{balanced-model}.} Two-dimensional simulations of the balanced model \eqref{balanced-model} on a rectangular domain of size $[0, 50]\times[0,50]$ for the transition and growth functions \eqref{Func_Pham}. The first row shows the evolution of the population $u(t,x)$, while the second row illustrates $v(t,x)$, with $x=(x_1,x_2)$. The three columns illustrate the population behavior at $t=0,100,500$, respectively. The regular mesh size is set to $\Delta x_1=\Delta x_2=0.2$, while the parameter values are chosen as $\nu=4$, $K=1$, $r=1$, $\gamma=1.6$, $n^*=0.75$, and $d=1$.}
    \label{Turing1}
\end{figure}
To further analyze the impact of the discretization, we illustrate in Figure \ref{Turing2} the final state of the $v$ variable when we reduce or enlarge $\Delta x$. We clearly see that the solution depends on the choice of $\Delta x$. We admit that normally, when a solution depends on the discretization, it implies that the numerical method is not good enough. But as we describe earlier, this is not due to an inappropriate numerical solver, rather, it is a systematic problem of this kind of models. 
\begin{figure}[!h]
    \centering
    \includegraphics[width=\linewidth]{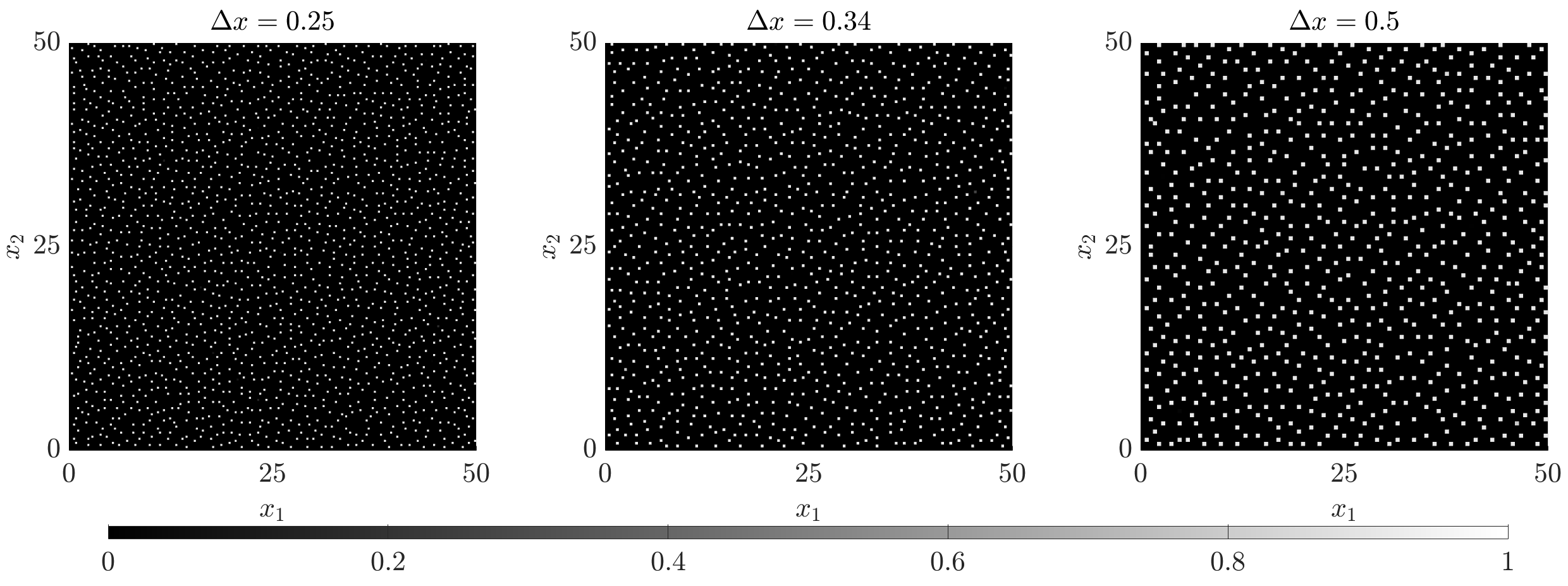}
    \caption{{ \bf Example of 2D instabilities emerging from model \eqref{balanced-model}.} Two-dimensional simulations at time $t=500$ of the population $v(t,x)$, with $x=(x_1,x_2)$, for the balanced model~\eqref{balanced-model} for three increasing value of the mesh size: $\Delta x=0.25, 0.34, 0.5$. Here $\Delta x=\Delta x_i$, for $i=1,2$. The model is simulated on a rectangular domain of size $[0, 50]\times[0,50]$ for the transition and growth functions \eqref{Func_Pham}. Parameters values are set to $\nu=4$, $K=1$, $r=1$, $\gamma=1.6$, $n^*=0.75$, and $d=1$. }
    \label{Turing2}
\end{figure}

\subsection{The go-or-grow case} 
Let us now employ the above results for go-or-grow models. Turing patterns for go-or-grow were previously considered by Pham et al.~\cite{pham2012} for the balanced model \eqref{balanced-model}. The function $g(n)$ is assumed to be non-increasing ($g'\leq 0$ with $g(1)=0$). In this case, we have two homogeneous steady states, $(0,0)$ and $(\Gamma(1), 1-\Gamma(1))$. We assume that $(0,0)$ is linearly unstable, such that a tumor can indeed grow, and consider diffusion driven instabilities at $(\Gamma(1), 1-\Gamma(1))$. Checking the autocatalysis condition \eqref{autocatalysis}, a necessary and sufficient condition for diffusion driven instability is 
\begin{equation}\label{Phamcond}
\gamma(\Gamma(1) +\Gamma'(1)) -g'(1)(1-\Gamma(1)) < 0.
\end{equation}
This condition was derived by Pham et al. \cite{pham2012} for the special choice of $g(n)=1-n$.  Note that since $g$ is assumed to be non-increasing, an increasing probability function $\Gamma(n)$ can never satisfy \eqref{Phamcond}. However, a decreasing function $\Gamma(n)$ with $|\Gamma'(n)|$ large enough can satisfy \eqref{Phamcond}.

Performing a linear stability analysis of the balanced model \eqref{balanced-model}, we find that under condition \eqref{Phamcond} all unstable modes $\omega$ satisfy \eqref{omegacond}, which in this context becomes 
\begin{equation}\label{unstablemodes}
\omega^2 > \frac{\gamma g'(1) (1-\Gamma(1))}{\gamma(\Gamma(1)+\Gamma'(1)) - g'(1) (1-\Gamma(1))}.
\end{equation}
Note that we use again the mode $\omega$ as a continuous variable here. If we consider the balanced model \eqref{balanced-model} on an interval $[0,L]$ with homogeneous Neumann or Dirichlet boundary conditions, then the modes become $\omega_n =n\pi/L$. 

For the probability of switching from stationary to moving \cite{pham2012}, there are two choices:
\begin{itemize}
\item \textit{Attractive case:} $\Gamma(n)$ is a decreasing smooth step function connecting $\Gamma(n)\approx 1$ for small $n$ to  $\Gamma(n)\approx 0$  for large $n$. Specifically, 
\[ \Gamma(n) = \frac{1}{2} ( 1+\tanh (\nu (n^*-n))), \]
where $n^*$ is the population size such that $\Gamma(n^*)=0.5$. In this case, cells are more motile at low densities and more stationary at high densities.
\item \textit{Repulsive case:} $\Gamma(n)$ is an increasing smooth step function connecting $\Gamma(n)\approx 0$ for small $n$ to  $\Gamma(n)\approx 1$  for large $n$. Specifically, 
\[ \Gamma(n) = \frac{1}{2} ( 1-\tanh (\nu (n^*-n))). \]
Cells are more motile at high densities and less motile at low densities. 
\end{itemize}

\noindent In the repulsive case, we have $\Gamma'>0$ and condition (\ref{Phamcond}) cannot be satisfied, but in the attractive case, we have $\Gamma'<0$ and for sufficiently large $\nu$, condition \eqref{Phamcond} can be satisfied. Pham et al.~\cite{pham2012} show numerical simulations for the attractive case that leads to spatial pattern formation and irregular traveling waves. 
However, Pham et al.~\cite{pham2012} do not discuss the high level of instability in their numerics. In Figure \ref{Pham_1D}, we show our own simulations of the traveling invasion waves of Pham et al.~\cite{pham2012}. As in Figures~\ref{Turing1}--\ref{Turing2}, we observe instabilities behind the wave on the scale of the discretization. Hence, we can not ensure that the numerical solver is correct in those cases. 

\begin{figure}[!h]
    \centering
    \includegraphics[width=0.32\linewidth]{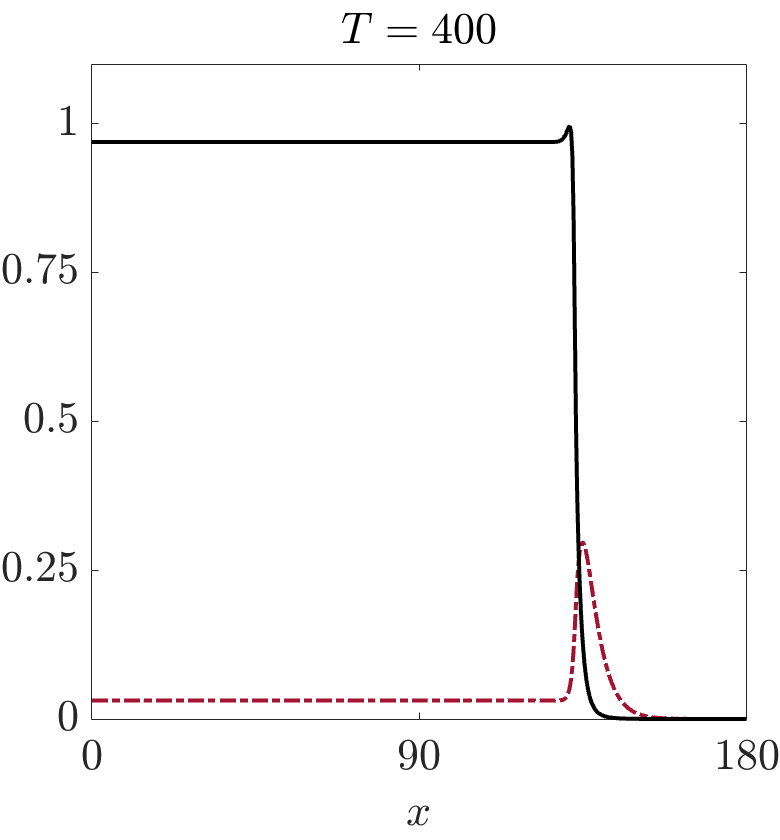}
    \includegraphics[width=0.32\linewidth]{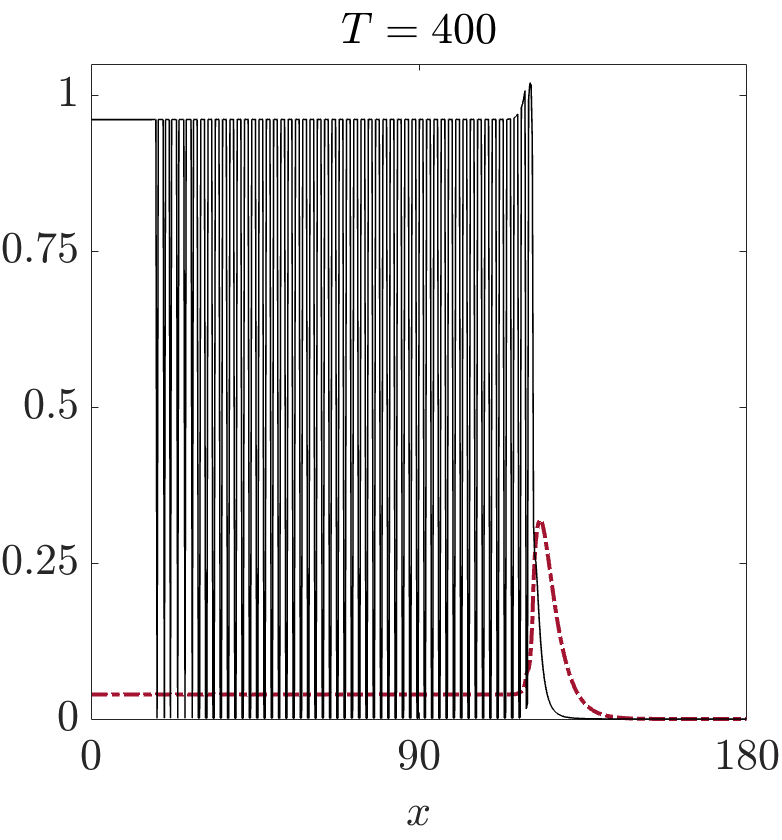}
    \includegraphics[width=0.32\linewidth]{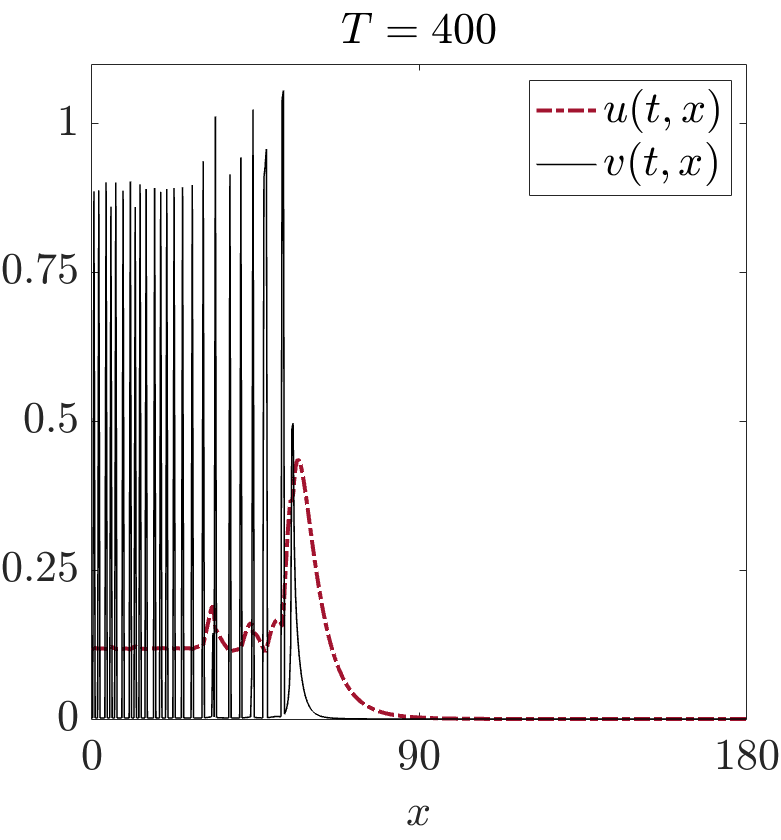}
    \caption{{\bf Examples of 1D instabilities from model \eqref{balanced-model}.} One-dimensional simulations of the balanced model \eqref{balanced-model} on the domain $[0, 180]$ for the transition and growth functions \eqref{Func_Pham}. The value of $n^*$ is increased from left to right: $n^*=0.57$, $n^*=0.6$, and $n^*=0.75$. The other parameters are set to  $\nu=4$, $K=1$, $r=1$, $\mu=1.6$, and $d=1$.}
    \label{Pham_1D}
\end{figure}

We have highlighted the rich array of intriguing mathematical properties related to the emergence of unusual behaviors and instabilities within the proposed framework. In the following, we focus on two well-established mathematical aspects for which more conclusive results can be obtained: the critical domain size problem and the traveling wave analysis.

\section{Critical domain size}\label{sec:criticaldomainsize}

A common focus in spatial ecology is to identify the minimal size of a habitat such that a species has enough resources to survive, known as the critical domain size \cite{deVries:2005:ACM,murray2003-II,britton1986reaction,Levin}. This has relevance in determining how large a habitat must be for species survival or how small it must be so that a parasite cannot establish itself. 

\subsection{The FKPP case}

We first review the analysis of the critical domain size for the FKPP equation~\eqref{FKPP} under homogeneous Dirichlet boundary conditions. We look for nontrivial steady states, which are nonnegative solutions of the boundary value problem,
\begin{equation}\label{FKPP_SS}
    \begin{sistem}
         -d\Delta m = g(m)m\qquad \rm{in}\,\Omega, \\[0.2cm]
         m = 0 \qquad\qquad\qquad \,\,\,\,\rm{on}\, \partial\Omega,
    \end{sistem} 
\end{equation}
where $\Omega \subset$ is a bounded domain with a smooth boundary $\partial \Omega$. We denote by $\lambda_1(\Omega)$ the value of the leading eigenvalue of $-\Delta$ on $\Omega$. This eigenvalue often plays an important role in growth phenomena on a given domain, and it will be the determining value here. In particular, there is a bifurcation in the behavior of solutions of \eqref{FKPP_SS} based on the size of the principal eigenvalue and the growth rate of the population at zero, i.e., $g(0)$.

\begin{theorem}[Theorems 4.60--4.61 in \cite{britton1986reaction}]\label{FKPP_CD}
If the growth function $g\in C^2(\mathbb{R})$ is of logistic type \eqref{logistic}, and satisfies the {\it subtangential condition} \eqref{st}, then we have two cases:
    \begin{enumerate}
        \item If $\lambda_1(\Omega) \leq \dfrac{g(0)}{d}$, then \eqref{FKPP_SS} has a nontrivial solution; 
        \item If $\lambda_1(\Omega) > \dfrac{g(0)}{d}$, then \eqref{FKPP_SS} only has the trivial solution.
    \end{enumerate}
\end{theorem}

This general result leads to the question of how the leading eigenvalue $\lambda_1$ depends on the domain size, which is clarified in the following theorem.

\begin{theorem}[Theorem 2.21 in \cite{wang2024upper}] \label{FKPP_CD_eigenvalue}
    For the eigenvalue problem associated with $-\Delta$, $\lambda_1(\Omega)$ is strictly decreasing with respect to $\Omega$, in the sense that if  $\,\Omega \subset \tilde{\Omega}$, then $\lambda_1(\Omega) \geq \lambda_1(\Tilde{\Omega})$, and  $\Omega \subsetneq \Tilde{\Omega}$ implies $\lambda_1(\Omega) > \lambda_1(\Tilde{\Omega})$.   
\end{theorem} 

The dependency on the domain size may furthermore be made more explicit by looking at the one-dimensional setting. Setting  $\Omega = [0,L]$, then the eigenvalue is $\lambda_1=(\pi/L)^2$, and we can rewrite Theorem~\ref{FKPP_CD} in the following manner.

\begin{corollary}
    If $\Omega = [0,L]$, $g\in C^2(\mathbb{R})$, and $g$ is of logistic type \eqref{logistic} and satisfies the {\it subtangential condition} \eqref{st}, then we have two cases:
    \begin{enumerate}
        \item If $L \geq \pi \sqrt{\dfrac{d}{g(0)}}$, then \eqref{FKPP_SS} has a nontrivial solution;
        \item If $L < \pi \sqrt{\dfrac{d}{g(0)}}$, then \eqref{FKPP_SS} only has the trivial solution.
    \end{enumerate}
\end{corollary}

\noindent In this one-dimensional setting, we have a clear bifurcation in the behavior of solutions based on the domain size $L$, and denote the bifurcation value for $L$ the \ul{\emph{critical domain size}}:
\begin{equation} \label{eq:FKPP_critdomain}
    L_{\small crit} = \pi \sqrt{\frac{d}{g(0)}}.
\end{equation}

\subsection{The go-or-grow case}

Returning to the go-or-grow model \eqref{general-model}, similarly to the previous section, we look for nontrivial steady states, which will be nonnegative solutions of the following boundary value problem  
\begin{equation}\label{criticalDomainProblem}
\begin{sistem}
     -d\Delta u = - \mu u -\alpha(u,v)u + \beta(u,v)v,\quad \rm{in}\, \Omega , \\[0.2cm]
    g(u+v)v +\alpha(u,v)u - \beta(u,v)v =0, \quad \rm{in}\, \Omega , \\[0.2cm]
    (u,v) = (0,0), \qquad \qquad \qquad\qquad \qquad\,\,\,  \rm{on} \, \partial\Omega ,
\end{sistem}
\end{equation}
where $\Omega$ is a bounded domain with a smooth boundary.

The critical domain size problem for the case where $\alpha$ and $\beta$ are constant (the constant rates go-or-grow model~\eqref{constantrates-model}) was studied in \cite{HadelerLewis2002}, where the following result was proven.

\begin{theorem}[Theorem 4 in \cite{HadelerLewis2002}] \label{HL_CD}
Consider (\ref{criticalDomainProblem}) with constant rates $\alpha$ and $\beta$. If $g\in C^2(\mathbb{R})$, and $g$ is of logistic type \eqref{logistic} and satisfies the {\it subtangential condition} \eqref{st}, then we have three cases:
\begin{enumerate}
    \item Assume $\beta > g(0)  $ and  
    \begin{equation}\label{condHL1}
    \lambda_1(\Omega) < \frac{1}{d}\left(\frac{\alpha g(0)}{\beta - g(0)} - \mu\right), 
\end{equation}
then there exists a nontrivial steady state  solution;
\item Assume $\beta > g(0)  $ and   
    \begin{equation}\label{condHL}
    \lambda_1(\Omega) >\frac{1}{d}\left(\frac{\alpha g(0)}{\beta - g(0)} - \mu\right) ,
\end{equation}
then the steady state is zero;
    \item If $\beta < g(0)$, then the zero solution is always unstable, and nontrivial steady states exist. 
\end{enumerate}
\end{theorem}

\noindent Similar to the FKPP case, we have a bifurcation that is based on the domain size in cases 1 and 2. The threshold in \eqref{condHL1}--\eqref{condHL} depends on the growth rate and, somewhat unsurprisingly, on the transition rates $\alpha$ and $\beta$ and the death rate $\mu$. What is surprising is that a third case appears for a small transition rate $\beta$. If the growth rate of the $v$ population at 0 is greater than the transition rate out of the $v$ compartment, i.e., $\beta < g(0) $, then there is no critical domain or bifurcation on the domain size. This means that for any size or shape of domain $\Omega$, there exist nontrivial populations of $u$ and $v$.

This property is another example of the unexpected properties of the go-or-grow systems. In fact, we can link this case to the autocatalysis condition \eqref{autocatalysis} discussed earlier. If we compute \eqref{autocatalysis} for the constant rates model \eqref{constantrates-model}, we find indeed
\begin{align*}
    \frac{\partial h}{\partial v}\bigg|_{(u,v) = (0,0)} = g(0) - \beta > 0. 
\end{align*}
In the autocatalysis case where $\beta - g(0) > 0$, we can directly relate the results from Theorem \ref{FKPP_CD} and Theorem \ref{HL_CD} for $\Omega = [0,L]$ and $\mu =0$. The threshold condition \eqref{condHL} can be written as a condition for the domain length $L$
\begin{equation}\label{domainsizeHL} 
    L> L_{\tiny\mbox{crit} } = \pi \sqrt{\frac{d(\beta-g(0)) }{g(0)\alpha }} = \pi\sqrt{\frac{d}{g(0)}}\underbrace{\sqrt{\frac{\beta - g(0)}{\alpha}}}_{K_{\alpha\beta}}.
\end{equation}
In this example, the critical domain of FKPP~\eqref{eq:FKPP_critdomain} and the critical domain of the constant rates go-or-grow model~\eqref{domainsizeHL} are related by the factor $K_{\alpha\beta}$. Setting $K_{\alpha\beta} = 1$ gives the equation when the two critical domains are equal, $\alpha + g(0) = \beta$. If we increase the transition rate to the moving population, $\beta$, or decrease the transition rate to the growing population, $\alpha$, from the $K_{\alpha\beta} = 1$ values, the critical domain of the go-or-grow model will be larger than the critical domain of FKPP. Similarly, if we reduce $\beta$ or increase $\alpha$,  the critical domain of the go-or-grow model will be smaller than the critical domain of FKPP. 

\section{Traveling wave analysis and invasion speeds}\label{sec:tw} 

To quantify the spread of populations in terms of the speed of invasion, traveling wave analysis has become a widely adopted approach, with a particular focus on determining the existence of traveling waves and their minimum speed. This is due to the fact that both the FKPP equation~\eqref{FKPP} and go-or-grow models~\eqref{general-model} inherently describe situations where waves of population density propagate at a constant speed. 

\subsection{The FKPP case} \label{sec:travwave_FKPP}
The analysis of traveling waves for the FKPP equation~\eqref{FKPP} is a well-researched and well-described textbook method \cite{murray2002-I,smoller,britton1986reaction,Levin,LewisLiWeinberger}. Consider \eqref{FKPP} on the real line and assume the growth function $f$ is of logistic type \eqref{logistic}, we introduce the {\it traveling wave coordinate} $z=x-ct$ with {\it wave speed} $c$, and define a traveling wave solution to be a self-similar solution of the form $m(t,x)=M(z)$. Substituting this ansatz into \eqref{FKPP} yields a second-order ODE for $M$ of the form
\begin{equation} \label{eq:FKPP_travwave_ODE}
- c M' = d M'' + g(M) M,
\end{equation}
with boundary conditions $M(-\infty)=1$ and $M(\infty)=0$. This ODE can be analyzed with phase plane methods, where a traveling wave of \eqref{FKPP} corresponds to a heteroclinic connection between two equilibria in the phase space. This heteroclinic connection depends on the choice of the speed $c$. As done in \cite{murray2002-I,smoller,britton1986reaction,Levin,LewisLiWeinberger}, it can be shown that, under the subtangential condition \eqref{st}, the minimum wave speed is 
\[  \bar c_{\text{\tiny\rm FKPP}} =2\sqrt{dg(0)}.\]

\subsection{The go-or-grow case} \label{sec:travwave_go-or-grow}
As in the FKPP case, traveling wave solutions for go-or-grow models can be seen as connections between two equilibria of the system. Ahead of the wave, we find the trivial equilibrium $(u,v)=(0,0)$, while at the back of the wave we typically see a coexistence equilibrium $(u^*, v^*)$. To ensure these states exist, we make the following assumption:
\begin{subequations}
\begin{align}
 \begin{split}
\bullet&\,\,g \,\text{is a continuous function that is differentiable at }\, 0\,\text{ and }\, 1,\text{ it has one root at } \,n = 1 ,\\
&\text{ and } g(n)n \leq g(0)n\,\text{ for }\,0 \leq n \leq 1 ;
\end{split}\label{ga}\\
\bullet&\,\,\mbox{there is only one non-zero solution }\,(u^*,v^*)\in \mathbb{R}^2_+\,\mbox{ of the system}\label{TW1}\\[0.2cm]
   &\hspace{2cm} \begin{sistem}
        0 = -\mu u^* - \alpha(u^*,v^*)u^* + \beta(u^*,v^*)v^*, \\[0.2cm]
        \mu u^* = g(u^*+v^*)v^*.
   \end{sistem}\nonumber
\end{align}
\end{subequations}
Assumptions \eqref{ga}-\eqref{TW1}  guarantees that ${\bf F}(0,0) = {\bf F}(u^*,v^*) = (0,0)$ in \eqref{vector}, allowing us to formalize the question of a traveling wave with the following definitions.

\begin{definition}
    Consider the general go-or-grow model \eqref{general-model} and introduce the \ul{\emph{traveling wave coordinate}} $z=x-ct$ for a fixed speed $c$. A solution of the form $u(t,x) = U(x-ct)$ and $v(t,x) = V(x-ct)$ is called a \ul{\emph{traveling wave solution with speed $c$}}. The traveling wave is called a \ul{\emph{wavefront}} if $\big(U(z),V(z)\big)$ has limits
    \begin{equation*}
        \lim_{z\to + \infty} \big(U(z),V(z)\big) =(0,0), \quad 
        \lim_{z\to - \infty} \big(U(z),V(z)\big) =(u^*,v^*).
    \end{equation*}
\end{definition}
\noindent If $c>0$, the invading traveling wave moves to the right, while if $c<0$, the invading traveling wave moves to the left. 

There are two distinct notions for the wave speed, $c$, which are relevant in our analysis.

\begin{definition}
\mbox{}
\begin{itemize} 
\item The \ul{\emph{nonlinear spread rate}}, $c^*$, denotes the observed minimal spreading speed of solutions of the full nonlinear system \eqref{general-model}.
\item The \ul{\emph{linear spread rate}}, $\bar c$, denotes the spreading wave speed of the solutions of its linearization at the leading edge (see  \cite{HadelerLewis2002}). 
\end{itemize}
Specifically, we say that the system is \ul{\emph{linearly determined}} if $c^*=\bar{c}$.
\end{definition}

The first results for the existence of wave fronts and the expression for the corresponding spreading speed were obtained for the constant rates go-or-grow model~\eqref{constantrates-model} by Hadeler, Lewis, and collaborators~\cite{lewis1996biological,HadelerLewis2002}. The connection between nonlinear spread speed, $c^*$, and linear spread speed, $\bar{c}$, was established in the case that the function $g(u+v)$ only depends on the stationary population $g(v)$.

For the balanced go-or-grow model~\eqref{balanced-model} with $g(n)=1-n$, Pham et al.~\cite{pham2012} showed numerically that invasion waves exist with different qualitative behavior. In particular, the wake of the wave might be a homogeneous steady state, a periodic wave train, or even an irregular pattern, as we have shown in Section \ref{sec:instabilities}.

More recently, Falc{\'o} et al.~\cite{falco2024traveling} examined the minimum invasion speed in the case of the total population go-or-grow model~\eqref{totalpop-model} with no death rate ($\mu = 0$), logistic growth for the proliferating population ($g(n) = 1-n$), and $\alpha(n)$ and $\beta(n)$ as nonincreasing and nondecreasing functions, respectively. By assuming $\alpha(n) > 0$ for all $n >0$, they ensure the existence of only two equilibria in the system,
\begin{align*}
    E_0 = (0,0),\qquad E_1 = \left(\frac{\beta(1)}{\alpha(1)+ \beta(1)}, \frac{\alpha(1)}{\alpha(1) +\beta(1)}\right),
\end{align*}
and find the traveling wave solutions as heteroclinic connections between these two equilibria in phase space. Moreover, considering the linearized version of \eqref{totalpop-model} about the equilibrium $E_0$ and transforming to wave coordinates, they obtain the expression for the linear spread rate
\begin{equation}\label{FalcoLinearspeed}
  \bar{c} =  \inf_{\rho > 0}\frac{\rho^2 + 1-\alpha(0) - \beta(0) + \sqrt{\big(\rho^2 + 1-\alpha(0) -\beta(0)\big)^2 - 4\big(1-\beta(0)\big)\rho^2 + 4\alpha(0)}}{2\rho}\,. 
\end{equation}
Although this formula still includes the infimum operation, it is observed that
\begin{align*}
    \left. \dfrac{dc}{d\rho}(\rho) \right|_{\rho = 1} = \frac{\beta(0)- \alpha(0)}{\alpha(0) + \beta(0)} ,
\end{align*}
and thus, some special cases emerge.
\begin{itemize}
    \item For $\alpha(0) = \beta(0)$, \eqref{FalcoLinearspeed}  obtains its minimum at $\rho = 1$, which implies that $\bar{c} = 1$. 
    \item For $\beta(0) = 0$, it is possible to explicitly solve \eqref{FalcoLinearspeed}, obtaining 
\begin{align*}
    \bar{c} = \frac{1}{\sqrt{1+\alpha(0)}}. 
\end{align*}
In this case, if $\beta(n)\ne 0$ for $n > 0$, it has been demonstrated that ensuring the positivity of traveling wave solutions requires a wave speed  $c>\bar{c}$. Therefore, for the case where $\beta(0) = 0$, the traveling wave solutions exhibit a wave speed greater than $\bar{c}$ meaning that the wave speed is nonlinearly determined.
\item Similarly, if $\alpha(0) =0 $,
\begin{equation*}
    \bar{c}=
    \begin{sistem}
    \sqrt{1 - \beta(0)} , \quad \text{if}\,\,\, \beta(0) < 1,\\[0.2cm]
    0, \quad\quad\quad\quad\quad\,\, \text{otherwise} .
    \end{sistem}
\end{equation*} 
\item When $\alpha(0) \neq \beta(0)$, $\bar{c} \leq c(1)=1$.
\end{itemize}

Similar results for the total population go-or-grow model \eqref{totalpop-model} were obtained in Stepien et al.~\cite{stepien2018} in the case $\mu\ne0$, assuming logistic growth with constant rate, $r$, and the following specific forms for the transition functions,
\begin{align} \label{stepienTF}
    \alpha(n)=\kappa\dfrac{K_v^k}{n^k+K_v^k}, \qquad \beta(n)=\epsilon\kappa\dfrac{n^k}{n^k+K_u^k} ,
\end{align}
with $\epsilon$, $\kappa$, $k$, $K_u$, and $K_v$ positive constants. A closed form for the dispersion relation was obtained, which provides the minimum wave speed
\begin{equation}\label{Tracyc}
\bar{c}=r\sqrt{\dfrac{d}{r+\kappa+\mu}}\,.
\end{equation}
Moreover, the existence of a heteroclinic orbit in phase space that corresponds to the approximate traveling wave (under certain conditions) was proven.

\subsection{A unifying framework and a general result} \label{sec:travwave_unifyingframework}
The previous results can be cast into the framework of cooperative reaction--diffusion systems. The lemmas and theorems presented here are new, and they are based on a general theory of cooperative reaction--diffusion systems developed in  \cite{liang2007asymptotic, fang2009monotone,weinberger2002analysis}. The proofs require a straightforward check of the general assumptions, which is presented in Appendix \ref{Proof-section}. 

Before we show the existence of traveling waves, we first characterize the wave speed $\bar c$  through the linearization at the leading edge. For this, we make the following assumptions: 
{\allowdisplaybreaks
\begin{subequations}
\begin{align}
\begin{split}\label{TW2}
\bullet\,\,& \alpha(u,v)\; \mbox{and } \beta(u,v)\;\mbox{are piecewise continuously differentiable in}\; (u,v)\; \mbox{for}\; {0\leq u\leq u^*}\\&\mbox{and}\; 0\leq v\leq v^*\mbox{, differentiable at}\;(0,0) \mbox{, and}\; \hspace{7.5cm} \\[0.1cm]
&\hspace{2.5cm}(g(0)-\beta(0,0))(\mu +\alpha(0,0)) +\alpha(0,0)\beta(0,0) >0; 
\end{split}\\
\begin{split}\label{TW3}
\bullet\,\,& \mbox{For any}\; (u,v)\in [0,u^*]\times [0,v^*], \hspace{7.5cm} \\[0.1cm] 
 &\hspace{2.5cm} \alpha(u,v) + \frac{\partial}{\partial u}\big(\alpha(u,v)\big) u + \frac{\partial}{\partial u}\big(g(u+v) - \beta(u,v)\big)v \geq 0, \\
 & \mbox{and} \\[0,1cm]
 &\hspace{2.5cm} \beta(u,v) +\frac{\partial}{\partial v}\big(\beta(u,v)\big) v - \frac{\partial}{\partial v}\big(\alpha(u,v)\big)u \geq 0.
 \end{split}
    \end{align}
\end{subequations}}
Assumption \eqref{TW2} gives regularity such that we can take the Jacobian at $(0,0)$ and ensures that the Jacobian at $(0,0)$ has a positive eigenvalue. Note that, if $\mu = 0$, the condition ${(g(0)-\beta(0,0))(\mu +\alpha(0,0)) +\alpha(0,0)\beta(0,0) >0}$ is satisfied. The assumption $\eqref{TW3}$  is sufficient to ensure that \eqref{general-model} is cooperative. 

We transform system (\ref{vector}) the into wave coordinate $z = x -ct$ and linearize  about $(0,0)$ to obtain
\begin{align}\label{wavesystem}
    -c\mathbf{w}' = D\mathbf{w}'' + \textbf{F}'(0,0)\mathbf{w}, 
\end{align}
where ${\bf w} =(u,v)^T$, $D=\mbox{diag}(d, 0)$,  $\textbf{F}$ is as defined in (\ref{F_vec}), and $\textbf{F}'(0,0)$ is the Jacobian of ${\bf F}$ at $(0,0)$.  
We write the solution of (\ref{wavesystem}) as  $\mathbf{w} = \bnu e^{-\rho z}$, where $\rho$ denotes the exponential decay rate at the leading edge and $\nu$ is an arbitrary vector. Substituting this ansatz into \eqref{wavesystem} yields the algebraic system 
\begin{align*}
    -c\rho \nu = D\rho^2 \nu +  \textbf{F}'(0,0)\nu.
\end{align*}
Defining the matrix $A(\rho) = D\rho^2 + \textbf{F}'(0,0)$, we arrive at the eigenvalue problem given by
\begin{align*}
    (A(\rho) - c\rho\mathbb{I})\nu =0 ,
\end{align*}
for the eigenvalue $\lambda_A=c\rho$. 
We find that the principal and positive eigenvalue of $A$ is 
\begin{multline} \label{eq:principal_eig_A}
    \lambda_A(\rho) = \frac{d\rho^2 + g(0) - \mu -\alpha(0,0) - \beta(0,0)}{2} \\
    + \frac{\sqrt{\big[d\rho^2 + g(0)- \mu -\alpha(0,0) - \beta(0,0)\big]^2 + 4\big[(g(0) -\beta(0,0))(\mu +\alpha(0,0)-d\rho^2)+ \alpha(0,0)\beta(0,0)\big]}}{2},
\end{multline}
with associated eigenvector  
\begin{align}
    \nu(\rho)  =\left(\frac{\lambda_A(\rho) - g(0) + \beta(0,0)}{\alpha(0,0)}, \; 1\right)^T,\label{eq:principal_eigvector_A}
\end{align}
which is a positive vector for all $\rho \geq 0$. To find the minimum wave speed of the linearized system, $\bar c$, we take the minimum ratio of the eigenvalue and decay rate as 
\begin{align} \label{barc}
    \bar c := \inf_{\rho>0}\;   \frac{\lambda_A(\rho)}{\rho} .
\end{align}
We verified that, under the assumptions $\mu=0$, $\alpha=\alpha(n)$, and $\beta=\beta(n)$, condition \eqref{barc} coincides with \eqref{FalcoLinearspeed} derived in \cite{falco2024traveling} for the total population model \eqref{totalpop-model}. Furthermore, when $\mu\ne0$ and the transition functions are given by \eqref{stepienTF}, condition \eqref{barc} matches \eqref{Tracyc}, which was previously derived in \cite{stepien2018}.

\begin{lemma}\label{phiProp}
Let $\alpha$, $\beta$, and $g$ satisfy \eqref{TW1}, \eqref{TW2}, and \eqref{TW3} and consider the initial condition $\phi = (\phi_1,\phi_2)\in \mathcal{C}_{u^*,v^*} = \{(\phi_1,\phi_2) \in C(\mathbb{R};\mathbb{R}^2): 0\leq \phi_1(z) \leq u^*,\; 0 \leq \phi_2(z) \leq v^*\; \forall z\in \mathbb{R}\}$. Let  ${\mathbf{w}(t,x;\phi):= \big(u(t,x;\phi_1),v(t,x;\phi_2)\big)}$ be the unique solution of \eqref{general-model} through $\phi$.
Then, there exists a real number $c^*$ with $0< c^* \leq \bar{c}$, such that the following statements are valid:
\begin{enumerate}
    \item If $\phi$ has compact support, then $\lim\limits_{\substack{t\rightarrow \infty\\|x|\geq ct}}  \mathbf{w}(t,x;\phi)  = 0,$ for all $c > c^*$;
    \item For any $c\in(0,c^*)$ and $\mathbf{r}>\mathbf{0}$, there is a positive number $R_{\mathbf{r}}$ such that for any $\phi \in \mathcal{C}_{u^*,v^*}$ with $\phi \leq \mathbf{r}$ on an interval of length $2R_{\mathbf{r}}$, there holds $\lim\limits_{\substack{t\rightarrow \infty\\|x|\leq ct}} \mathbf{w}(t,x;\phi)  =(u^*,v^*)$;
    \item If, in addition, we have a linear bound for any $\eta >0$, i.e.,
    \begin{equation}\label{TW5} 
    g\big(\eta(1+\nu_1(\rho))\big) \leq-\alpha\big(\eta \nu_1(\rho),\eta\big)\nu_1(\rho) + \beta\big(\eta\nu_1(\rho),\eta\big) + \lambda_A(\rho)\leq g(0),
    \end{equation}
    for all $\rho \in (0,\rho^*]$, then $c^*= \bar{c}$, where $\rho^*$ is the value of $\rho$ at which $\lambda_A(\rho)/\rho$ attains its infimum, and $\nu_1(\rho)$ is the first component of the eigenvector \eqref{eq:principal_eigvector_A}.
\end{enumerate}
\end{lemma}
\begin{proof}
    The proof is based on the general theory of \cite{liang2007asymptotic, fang2009monotone,weinberger2002analysis} and is presented in Appendix~\ref{Proof-section}.
\end{proof}

Lemma~\ref{phiProp} provides sufficient conditions for the existence of a constant spreading speed of a compact initial condition in items 1 and 2. Item 3 gives conditions such that the model is linearly determined and that the nonlinear spread rate, $c^*$, coincides with the linearized spread rate, $\bar c$. However, to prove the existence of traveling waves, we also require condition \eqref{TW5} as well as the following assumption:
\begin{equation}
    \begin{split}\label{TW4}
        \bullet \,\,& \mbox{There exist}\; a>0,\; \sigma > 1,\; \mbox{and}\; {\bf r} > {\bf 0}\;\mbox{such that}\\[0.1cm]
        &\hspace{2cm}  \big(\alpha(0,0) - \alpha(u,v)\big)u + \big(\beta(u,v)-\beta(0,0)\big)v \geq -a|u + v|^{\sigma}, \\[0.1cm]
        &\hspace{2cm}\big(g(u+v) -g(0)\big)v -\big(\alpha(0,0) - \alpha(u,v)\big) u -\big(\beta(u,v)-\beta(0,0)\big)v \geq - a|u+v|^{\sigma}, \\[0.1cm]
        &\mbox{for all}\; {\bf 0}\leq (u,v) \leq \mathbf{r}.
    \end{split}
\end{equation}
Although condition \eqref{TW4} looks quite complex, it is already satisfied when ${\bf F}(u,v)$ is twice differentiable in both components \cite{fang2009monotone}. The assumption \eqref{TW5} comes from a variant of the subtangential condition \eqref{st}, requiring that ${\bf F}(u,v) \leq {\bf F}'(0,0)$ along the direction of the eigenvector associated with the principal eigenvalue $\lambda_A(\rho)$. With assumptions \eqref{TW1}, \eqref{TW2}--\eqref{TW3}, and \eqref{TW4}--\eqref{TW5}, we can show that \eqref{general-model} has traveling wave solutions.

\begin{theorem} \label{general-model-cooperative-TW}
    If the general grow-or-grow model~\eqref{general-model} satisfies the conditions \eqref{TW1}, \eqref{TW2}--\eqref{TW3}, and \eqref{TW5}--\eqref{TW4}, then for each $c < \bar{c}$, it has a nonincreasing wavefront $\big(U(x-ct),V(x-ct)\big)$ connecting $(u^*,v^*)$ to $(0,0)$. However, for any $c \in(0,\bar{c})$, there is no wavefront \mbox{$\big(U(x-ct),V(x-ct)\big)$} connecting $(u^*,v^*)$ to $(0,0)$.
\end{theorem}
\begin{proof}
    The proof requires a check that the conditions for the general theory of cooperative systems are satisfied, and it is presented in Appendix~\ref{Proof-section}.
\end{proof} 

Theorem~\ref{general-model-cooperative-TW} allows us to establish the existence of traveling wave solutions with a formula for the wave speed \eqref{barc} for a large class of go-or-grow models. Although many go-or-grow models fall into the above framework, there exists plenty of examples that fail one of the assumptions, such as those studied in \cite{stepien2018,tursynkozha2023traveling,falco2024traveling} that numerically display traveling wave solutions, but they fail the cooperativity condition \eqref{TW3}. To prove traveling wave solutions in those cases, one can attempt to construct upper and lower solutions and use a fixed point proof, as done in \cite{zhang2017minimal,zhang2016minimal,zhang2014existence}, or use a geometric shooting method, as in \cite{li2008traveling,hsu2012existence,huang2003existence,lin2011traveling,huang2012traveling,huang2016geometric}. Another common method is to use approximations, such as the fast transition rate scaling introduced in Section~\ref{sec:fast_trans_rate}, where the general go-or-grow model~\eqref{general-model} can be transformed into a FKPP equation \eqref{FKPP} for the total population depending on the form of the transition functions $\alpha$ and $\beta$. Since the FKPP equation \eqref{FKPP} has traveling wave solutions and a simple formula for the spreading speed, the fast-transition scaling can provide an approximation on the spreading speed of \eqref{general-model}.

\subsection{Connection between the spread rate of the FKPP equation and go-or-grow models}
To link the foundational analysis of the FKPP equation~\eqref{FKPP} as summarized in Section~\ref{sec:travwave_FKPP} with analysis of go-or-grow types of models~\eqref{vector} (Sections~\ref{sec:travwave_go-or-grow}--\ref{sec:travwave_unifyingframework}), we compare the spread rate between the cases.

\begin{theorem}
        Assume $\mu=0$, ${\bf F}({\bf w})$ is given by \eqref{F_vec}, and $\alpha$, $\beta$, and $g$  satisfy assumptions \eqref{logistic},  \eqref{st}, \eqref{TW1}, \eqref{TW2}--\eqref{TW3}, and \eqref{TW5}--\eqref{TW4}. Then the minimum wave speed, $\bar c$ of the go-or-grow model \eqref{general-model} given by \eqref{barc} satisfies the inequality
        \begin{equation}\label{minwave}
            \bar{c} \leq  \frac{1}{2} \bar{c}_{\text{\tiny\rm FKPP}} = \sqrt{d g(0)}.
        \end{equation}
\end{theorem}

\begin{proof}
Since ${\bf F}(u,v)$ satisfies assumptions \eqref{TW1}, \eqref{TW2}--\eqref{TW3}, and \eqref{TW5}--\eqref{TW4}, Theorem \ref{general-model-cooperative-TW} can be applied, and
we have that the minimum wave speed of the corresponding go-or-grow model is given by 
\begin{equation}
    \bar{c} = \inf_{\rho \geq 0}\frac{\lambda_+(\rho)}{\rho} ,
\end{equation}
where $\lambda_+(\rho)$ is the principal eigenvalue of the matrix 
\begin{equation}
    A = \begin{pmatrix}
        d\rho^2 & 0 \\
        0 & 0 
    \end{pmatrix}  + 
    \begin{pmatrix}
        -\alpha(0,0)  &+ \beta(0,0) \\
        \alpha(0,0) &- \beta(0,0) + g(0) 
    \end{pmatrix} .
\end{equation}
This can be written explicitly as  
\begin{multline} \label{eq:spreadspeed_principleeig}
    \lambda_+(\rho) = \frac{d\rho^2 + g(0) -\alpha(0,0)-\beta(0,0)}{2} \\ + \frac{\sqrt{\big(d\rho^2 + g(0) -\alpha(0,0)-\beta(0,0)\big)^2 -4\big(d\rho^2(g(0) - \beta(0,0))-\alpha(0,0) g(0)\big)}}{2} .
\end{multline}
We notice that substituting the choice of $\rho = \sqrt{g(0)/d}$ into \eqref{eq:spreadspeed_principleeig}, we get 
\begin{equation*}
    \frac{\lambda_+\left(\sqrt{\frac{g(0)}{d}}\right)}{\sqrt{\frac{g(0)}{d}}} = \sqrt{dg(0)}.
\end{equation*}
Thus, we have the set of inequalities 
\begin{align}
    \bar{c} = \inf_{\rho \geq 0} \frac{\lambda_+(\rho)}{\rho} \leq \sqrt{dg(0)} \leq 2\sqrt{dg(0)} = \bar{c}_{\text{\tiny\rm FKPP}}.
\end{align}
\end{proof}
Some special cases of this result were already proven in the literature for the constant rates model~\eqref{constantrates-model} in \cite{Fedotov2008,lewis1996biological} and for the total population model \eqref{totalpop-model} in \cite{falco2024traveling}. However, the results provided here are a generalization for the case of the general go-or-grow model \eqref{general-model}. As an example, in Figure \ref{fig:twicespeed} we illustrate the traveling wave profile for the constant rates model~\eqref{constantrates-model} (black and red lines) and the FKPP model \eqref{FKPP} (blue lines). These numerical simulations show how the FKPP model is at least two times faster than the constant rates go-or-grow model.
\begin{figure}[h!]
    \centering
    \includegraphics[width=0.6\linewidth]{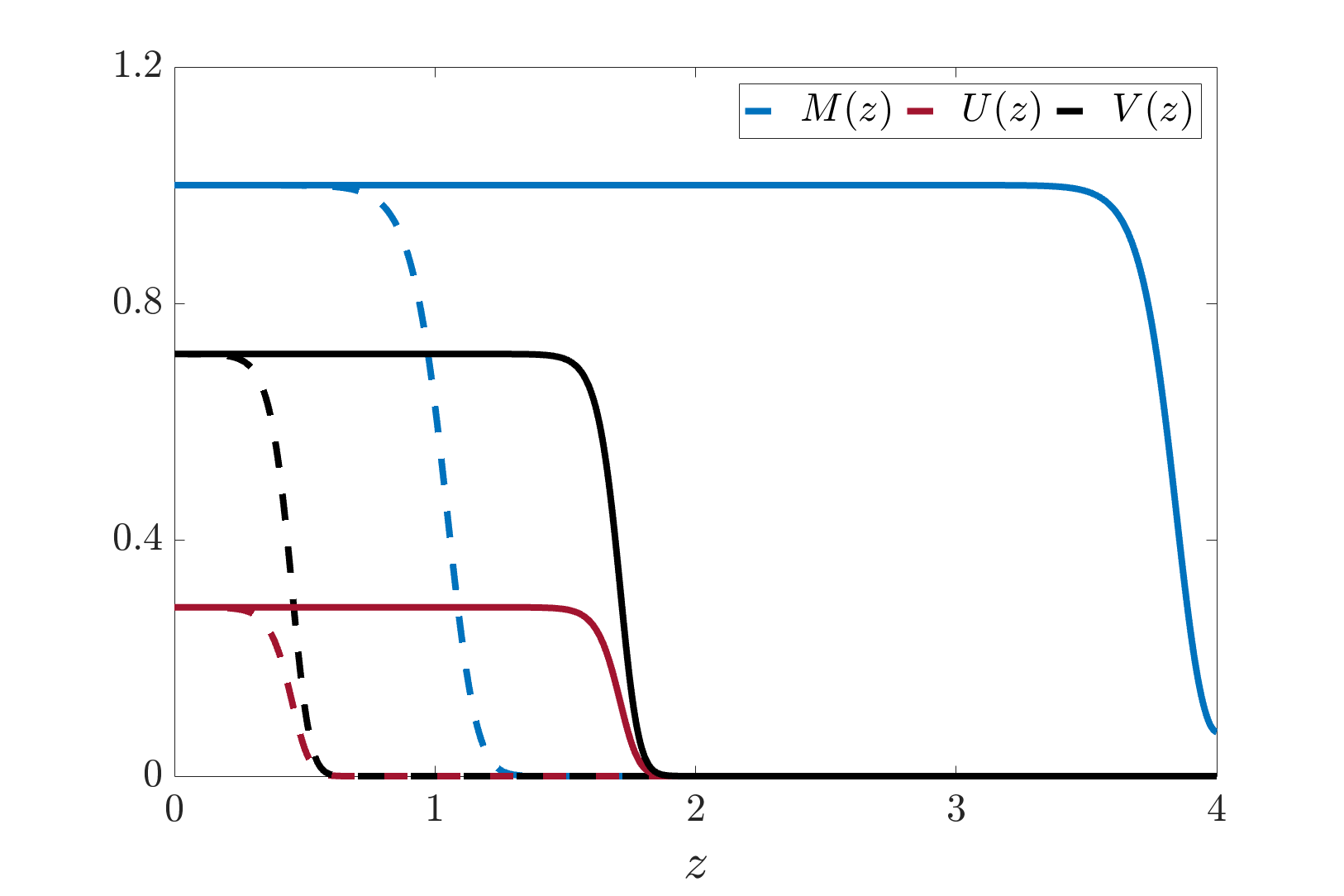}
    \caption{{\bf Traveling wave profiles for models \eqref{FKPP} and \eqref{constantrates-model}.} Traveling wave profiles for the FKPP model \eqref{FKPP} ($M(z)$ in blue) and the constant rate go-or-grow \eqref{constantrates-model} (moving cells $U(z)$ in red, proliferating cells $V(z)$ in black). Simulations are shown at two different time steps: an earlier time $T_1$ (dashed curves) and a later time $T_2$ (solid curves). Parameters are set to $\alpha = 1$, $\beta = 2$, $g(0) =1$, $\mu = 0$, and $d = 0.0001$.}
    \label{fig:twicespeed}
\end{figure}

\subsection{Large wave speed scaling}\label{sec:large_wave_speed_scaling}
In addition to studying the wave speed of traveling wave solutions of the general go-or-grow model~\eqref{general-model}, here we introduce the large wave speed scaling approach, which provides insights into the shape of the traveling waves. Originally developed by Canosa~\cite{canosa1973nonlinear}, this method allows for an approximation of the traveling wave solution by considering an asymptotic expansion of the traveling wave system in the limit of large wave speeds. For systems for which no general traveling wave solution theory exists, i.e., those systems which are noncooperative and not bounded above and below by cooperative systems, this technique is a useful way to obtain the form of the traveling wave by solving a much simpler approximate system. Furthermore, this approach facilities the investigation of the conditions under which the approximation aligns closely with the traveling wave system, providing a clearer understanding of the wave dynamics.

Introducing the traveling wave coordinate, $z= x-ct$, such that $U(z)=u(t,x)$ and $V(z)=v(t,x)$, the general go-or-grow model~\eqref{general-model} is written as
\begin{equation}\label{FWSS}
\begin{sistem}
    -cU' = DU''-\mu U -\alpha(U,V)U+\beta(U,V)V, \\[0.4cm]
    -cV' = g(U+V)V +\alpha(U,V)U-\beta(U,V)V.
    \end{sistem}
\end{equation}
By using the scaling $\xi = -\frac{1}{c}z$, we obtain the approximation
\begin{equation*}
\begin{sistem}
    U' = \dfrac{D}{c^2}U''-\mu U -\alpha(U,V)U+\beta(U,V)V , \\[0.4cm]
    V' = g(U+V)V +\alpha(U,V)U-\beta(U,V)V.
    \end{sistem}
\end{equation*}
Defining the perturbation parameter $\delta := 1/c^2 $ and assuming it is small, we expand $U$ and $V$ as regular perturbation expansions
\begin{equation*}
    U = \sum^{\infty}_{n=0}U_n\delta^n,\qquad  V = \sum^{\infty}_{n=0}V_n\delta^n. 
\end{equation*}
As in the fast transition rate scaling (Section~\ref{sec:fast_trans_rate}), we require that $\alpha$, $\beta$, and $g$ are sufficiently smooth to allow for a Taylor expansion at $(U_0,V_0)$. In addition, we require each order of $\delta$ to vanish independently. With both of these assumptions, the zeroth order of $\delta$ reads
\begin{equation}\label{gen-large-speed1} 
\begin{sistem}
    U_0' = -\mu U_0 -\alpha(U_0,V_0)U_0+\beta(U_0,V_0)V_0 , \\[0.4cm]
    V_0' = g(U_0+V_0)V_0 +\alpha(U_0,V_0)U_0-\beta(U_0,V_0)V_0.
    \end{sistem}
\end{equation}
\noindent The effective result of scaling is removing the diffusion term from the analysis, reducing the dimension of the phase space of the problem by one. In this case, we reduce the phase space to two dimensions, allowing for the examination of the shape of the traveling wave solutions.
\paragraph{Example} We illustrate the large wave speed approximation in the case of the constant rates go-or-grow model \eqref{constantrates-model} with $\mu = 0$ and logistic growth function, ${g(u+v)v = (1 -(u+v))v}$. In this example, system \eqref{gen-large-speed1} becomes  
\begin{equation}\label{constantrates-large-speed1} 
\begin{sistem}
    U_0' =  -\alpha U_0+\beta V_0 , \\[0.4cm]
    V_0' = (1 - U_0 -V_0)V_0 +\alpha U_0-\beta V_0.
    \end{sistem}
\end{equation}
The classical phase plane analysis shows that the origin (0,0) is a saddle node and the equilibrium ${(U_0^*,V_0^*) = \left(\frac{1}{\alpha / \beta +1},\frac{\alpha / \beta }{\alpha / \beta  +1} \right)}$ is a stable node. We calculate the nullclines of the above system in terms of $U_0$, obtaining
\begin{subequations}
\begin{align} \label{U_0nullcline}
U_0\mbox{-nullcline:}\qquad     V_{\mathcal{N}_1}(U_0) &= \frac{\alpha}{\beta}U_0, \\ 
 V_0\mbox{-nullcline:}\qquad    V_{\mathcal{N}_2}(U_0) &= \frac{(1-U_0) - \beta}{2} + \frac{\sqrt{((1-U_0) - \beta)^2+ 4\alpha U_0}}{2}.\label{V_0nullcline}
\end{align}
\end{subequations}
From these nullclines, it is possible to define the trapping region
\begin{align}
    \mathcal{W} = \{(U_0,V_0)\in \mathbb{R}^2_+: 0\leq U_0 \leq U_0^*,\;  V_{\mathcal{N}_1}(U_0) \leq V_0 \leq  V_{\mathcal{N}_2}(U_0)\}\,.
\end{align}
Indeed,  for trajectories on the line $\{U_0 = 0,\;V_0 > 0 \}$, we have $U_0' > 0$, thus they point inward into $\mathcal{W}$. Then, we consider the curve $\{V_0 =  V_{\mathcal{N}_1}(U_0),\; 0 <U_0 < U_0^*\}$, whose trajectories have $U_0' = 0$ and 
\begin{align*}
    V_0' = \frac{\alpha}{\beta}\left[1 - \left(1+\frac{ \alpha}{\beta}\right)U_0 \right]\,U_0  > 0,
\end{align*}
thus, the flux points upwards into $\mathcal{W}$. Finally, on the curve ${\{V_0 =  V_{\mathcal{N}_2}(U_0),\; 0 < U_0 < U_0^*\}}$, the trajectories have $V_0' = 0$ and 
\begin{align*}
    U_0' = -\alpha U_0 + \beta  V_{\mathcal{N}_2}(U_0) > -\alpha U_0 + \beta V_{\mathcal{N}_1}(U_0) = 0,
\end{align*}
meaning that the flux points into $\mathcal{W}$. Therefore, by Poincaré--Bendixson theorem (Theorem 3.7.1 in \cite{perko2013differential}), we have that every $\omega$-limit set of a trajectory inside $\mathcal{W}$ is either a fixed point, a periodic orbit, or a heteroclinic orbit between the two fixed points $(0,0)$ and $(U_0^*,V_0^*)$. Considering the Dulac function $\phi = -1/V$, we can rule out periodic orbits by applying Dulac's criterion (Theorem 3.9.2 in \cite{perko2013differential}). Thus, we obtain the existence of a heteroclinic orbit between $(0,0)$ and $(U_0^*,V_0^*)$. Figure \ref{fig:phasePortrait} illustrates the phase portrait of the system \eqref{constantrates-large-speed1} and the heteroclinic orbit between the two equilibria on the left, as well as the heteroclinic orbit and the traveling wave solution of \eqref{constantrates-large-speed1} in the wave coordinate $\xi$ on the right.

\begin{figure}[h!]
    \centering
    \includegraphics[width=0.43\linewidth]{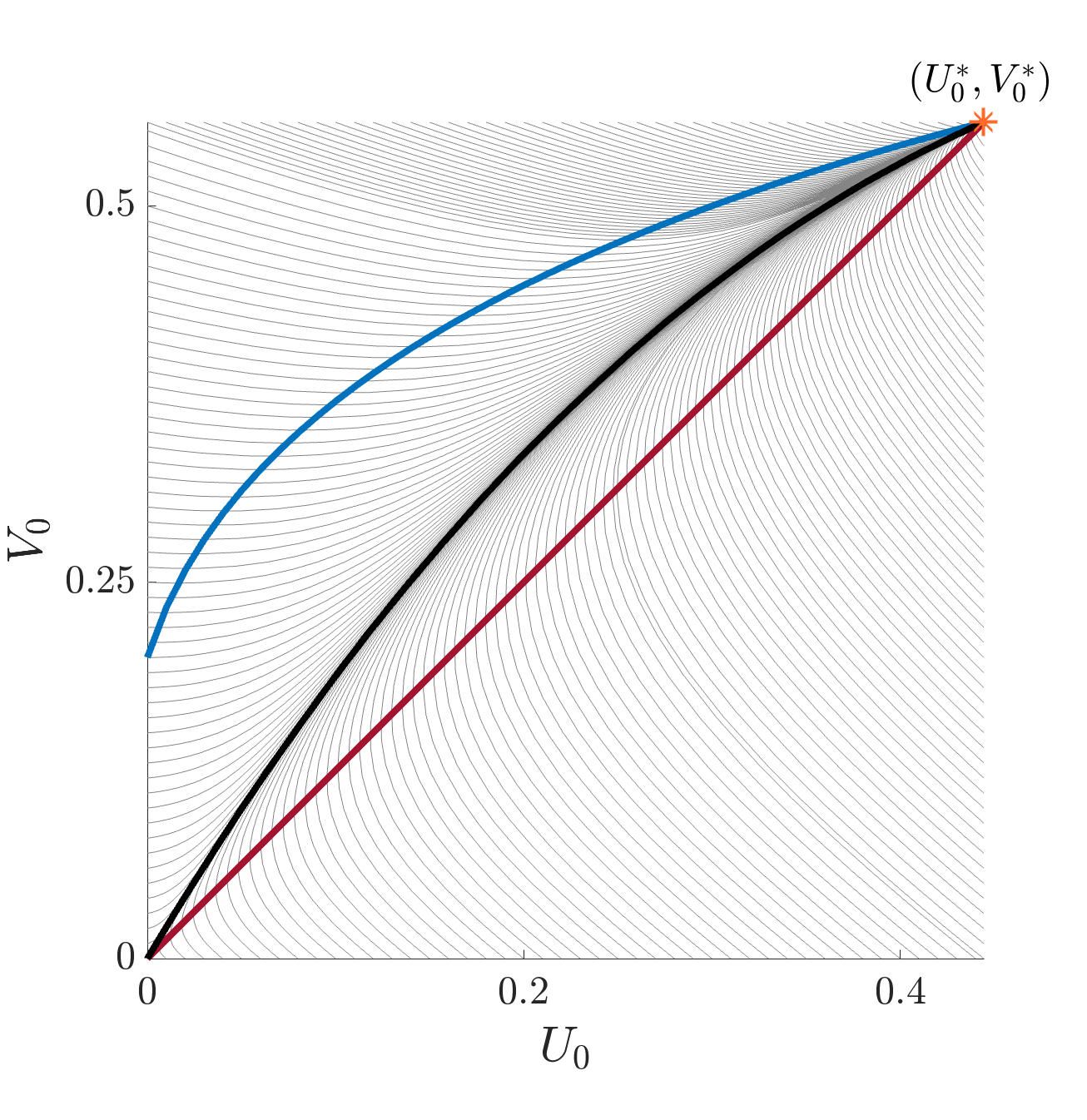}
    \includegraphics[width=0.43\linewidth]{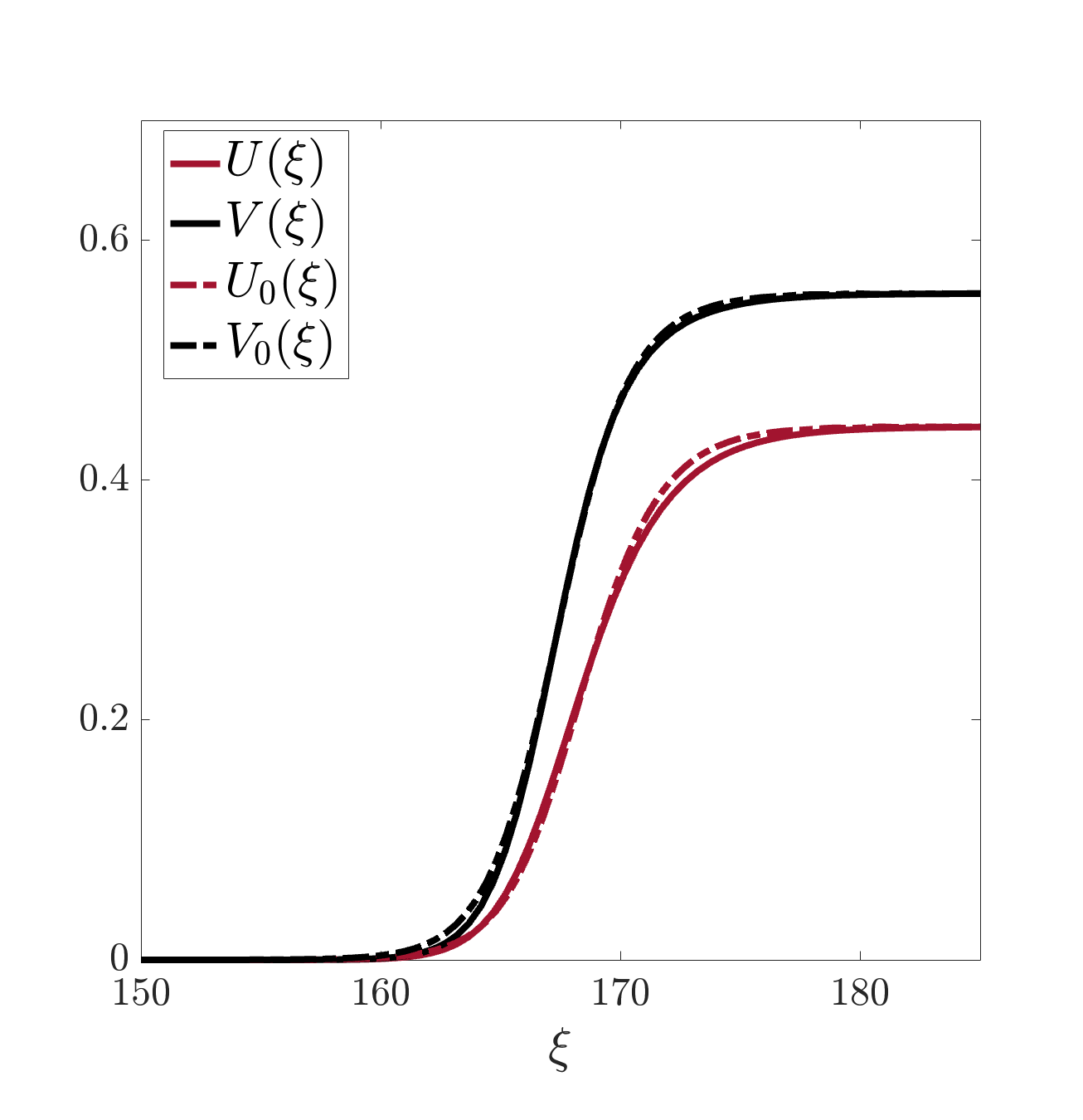}
    \caption{{\bf Large wave speed scaling applied to model \eqref{constantrates-model}.} (Left) Phase portrait of system \eqref{constantrates-large-speed1}. Trajectories are depicted in light grey, while the blue curve represents the $V_0$ nullcline \eqref{V_0nullcline}, the red curve represents the $U_0$ nullcline \eqref{U_0nullcline}, and the black curve represents the heteroclinic connection from the $(0,0)$ equilibrium to $(U_0^*,V_0^*)$ (orange marker). (Right) Comparison between the traveling wave profile (solid curves) from the model \eqref{constantrates-model} with the corresponding large wave speed approximation \eqref{constantrates-large-speed1} (dot-dashed curves) in the coordinate $\xi$. Logistic growth is assumed in \eqref{constantrates-model}, and the parameters (for both plots) are set to $\alpha =1$, $\beta =0.8$, $\mu=0$, and $d = 0.0001$.}
    \label{fig:phasePortrait}
\end{figure}

\bigskip
The large wave speed scaling has been used extensively in mathematical biology, especially in cancer growth models  \cite{sherratt2000wavefront,sherratt2001new,zhu2008justification,quinn2009dynamics,zhu2011existence}. It was explicitly used for the total population go-or-grow models in \cite{stepien2018,tursynkozha2023traveling} for obtaining an approximation of the traveling wave solution. In particular, in \cite{stepien2018}, the authors numerically show that for small values of the wave speed, the simulations do not agree, while for larger values of wave speed they observe that the approximation is valid.

\section{Discussion} \label{sec:discussion}
In this article, we have examined the significance of a class of go-or-grow models that are widely employed to describe various biological phenomena, particularly those related to cell dynamics and cancer progression. What started as a simple review of go-or-grow models quickly became a mathematically involved exposition. Our primary objective has been to provide a comprehensive review of the intriguing mathematical features that characterize this approach, while also highlighting the challenges and open problems that may inspire future mathematical investigations. Despite the biological relevance of these models and their extensive application across various mathematical contexts, a systematic exploration of their mathematical properties and the associated analytical and numerical challenges has been lacking. While we do not claim our analysis is exhaustive, we have endeavored to delve into the diverse array of analytical results available in the literature, examining different mathematical assumptions related to the switching terms of the models and establishing clear connections between the classical FKPP equation and its go-or-grow counterparts.

Beginning with a brief overview of the biological evidence for the emergence of dichotomic dynamics in biological systems, we focused specifically on glioma cell evolution. Research has demonstrated that the progression of this tumor is driven by a proliferation/migration dichotomy, where proliferating cells remain stationary or move slowly, while migrating cells do not proliferate \cite{giese1996,giese2003cost}. This serves as the primary biological motivation for our exploration of this class of models. We provided a comprehensive summary of the extensive mathematical literature addressing the go-or-grow dichotomy in glioma progression, spanning from stochastic and discrete approaches to continuous frameworks, with migratory dynamics and population switches influenced by various external and intrinsic factors.

We focused on a specific continuous modeling framework, which we named as the {\it general go-or-grow model} \eqref{general-model}. This model features linear diffusion for the moving population and a characteristic switch between the two populations: migrating cells ($u)$ and proliferating ($v$) cells. The transition functions $\alpha(\cdot)$ and $\beta(\cdot)$ govern this switching and depend solely on the densities of the populations. Depending on the nature of this dependency, we proposed various specialized versions of the model and reviewed the key papers that have contributed to their analysis. Moreover, we revisited a general existence result for the general model and discussed common scaling techniques that facilitate its analysis, particularly in relation to the classical FKPP equation. 

Although these go-or-grow models have been widely employed, they present a range of analytical and numerical challenges associated with the emergence of highly unstable patterns and irregular traveling waves under specific assumptions on the transition functions. While these issues have been observed in different contexts \cite{pham2012,Anna2}, we highlighted and integrated them into a common framework. Singularities arise, driven by an extreme version of the Turing instability, where the diffusion coefficient of the activator is zero. All high frequencies are unstable, leading to oscillations and patterns on fine scales. Of course, this instability has conditions, and if these conditions are not satisfied, then the model is well-behaved.

Focusing on specific analytical properties of go-or-grow models that can be effectively addressed, we extensively reviewed results related to the critical domain problem and the analysis of the existence and minimal speed of traveling wave solutions for the general go-or-grow model and its specialized variants. For the critical domain problem, we revisited classical results concerning the FKPP equation \cite{britton1986reaction,wang2024upper} and connected them to findings related to the general go-or-grow model, which show the bifurcation between domain size and the existence of non-trivial steady-state solutions, under specific assumptions about the transition and growth functions. In our examination of traveling wave solutions, we outlined the necessary assumptions that the general go-or-grow model~\eqref{general-model} must meet to be classified as a cooperative system. This classification allows us to apply a robust theoretical framework concerning the existence of traveling waves and invasion speed, as established in the literature \cite{fang2009monotone,weinberger2002analysis}. Similar to our analysis of the critical domain size, we establish a connection to the FKPP equation, thereby generalizing previous results to the context of the general go-or-grow model~\eqref{general-model} introduced here.

\paragraph{Open Problems}
While we have reviewed a wide array of analytical and numerical results concerning this special class of mathematical models, several important mathematical questions remain unresolved. One intriguing question that arises from our analysis pertains to the convergence of the fast-transition rate system \eqref{Scaling_FTR} to system \eqref{gen-fast-transition1} as $\ep \to 0$, specifically in the context of traveling wave solutions: what assumptions must be satisfied for the limit $\ep \to 0$ to ensure convergence (in some sense) of \eqref{Scaling_FTR} to \eqref{gen-fast-transition1}? Furthermore, if \eqref{Scaling_FTR} meets these necessary conditions, what does the wave speed of \eqref{Scaling_FTR} converge to in this limit? How does this relate to the minimum wave speed of the corresponding FKPP equation given by \eqref{FKPP_FTS}?

Similar questions arise when considering large wave speed scaling. Many biological systems can be described by models of the form \eqref{general-model}, with specific expressions for the transition functions $\alpha(u,v)$ and $\beta(u,v)$ that may result in noncooperative dynamics. Under what conditions is the existence of traveling waves guaranteed in these noncooperative models, and what is the minimum wave speed associated with them? For cooperative systems, when can we derive an explicit expression for the minimum wave speed, as indicated in \eqref{barc}?

Additional challenges emerge from the instabilities and patterns discussed in Section \ref{sec:instabilities}. How can the autocatalysis condition \eqref{autocatalysis} be formulated for a general model within the go-or-grow family? Is it feasible to develop more sophisticated numerical solvers capable of correctly identifying the sources of the observed instabilities?

These open questions present intriguing opportunities for future research. As we intended to showcase the usefulness of go-or-grow modeling in biology while highlighting interesting mathematical properties, this article is also a warning not to underestimate the monster on a leash.

\subsection*{Acknowledgements} {\small We are grateful to the Fields Institute in Toronto, who hosted us for a few weeks in Summer and Fall 2024 to work on this project. 
 RT acknowledges support from a Canadian Graduate Scholarship of the Natural Sciences and Engineering Research Council of Canada (NSERC) and Alberta Innovates Graduate Student Scholarship.
MC was supported by the National Group of Mathematical Physics (GNFM-INdAM) through the INdAM–GNFM Project (CUP E53C22001930001) ``From kinetic to macroscopic models for tumor-immune system competition'', and by the European Union - NextGenerationEU and the MUR-Italian Ministry of Universities and Research through the project PRIN 2022 PNRR {\it Mathematical modeling for a Sustainable Circular Economy in Ecosystems} (project code P2022PSMT7, CUP D53D23018960001). MC has been partially supported by the State Research Agency of the Spanish Ministry of Science and FEDER-EU, project PID2022-137228OB-I00 (MICIU/AEI /10.13039/501100011033); by Modeling Nature Research Unit, Grant QUAL21-011 funded by Consejería de Universidad, Investigaci\'on e Innovaci\'on (Junta de Andalucía). TLS acknowledges support from a National Science Foundation (NSF) grant DMS-2151566.
TH is supported through a discovery grant of the Natural Science and Engineering Research Council of Canada (NSERC), RGPIN-2023-04269. }

\subsection*{Conflict of interest}
We declare that TH is currently an Editor-in-Chief of the Journal of Mathematical Biology.
The research was conducted in the absence of any commercial or financial interests that could be construed as a potential conflict of interest.  

\subsection*{Data availability statement}
Data sharing is not applicable to this article as no data sets were generated or analyzed during the current study.

\begin{appendices}

\section{Appendix: Proofs of Lemma \ref{phiProp} and Theorem \ref{general-model-cooperative-TW}}\label{Proof-section}

The proofs of Lemma \ref{phiProp} and Theorem \ref{general-model-cooperative-TW} primarily require showing that the general go-or-grow model~\eqref{general-model} with the assumptions \eqref{TW1}, \eqref{TW2}--\eqref{TW3}, and \eqref{TW5}--\eqref{TW4} lives inside the general cooperative theory proposed by \cite{liang2007asymptotic, fang2009monotone,weinberger2002analysis}. In these works, the authors are concerned the $n$-equation reaction--diffusion system
\begin{equation} \label{super-general-model}
    {\bf w}_t = D\Delta {\bf w} + {\bf F}({\bf w}),
\end{equation}
where ${\bf w} = (w_1,\dots,w_n)$ and $D = {\rm diag}(d_1,\dots,d_n)$, with $d_i\geq 0$ and at least one $d_i \neq 0$. Moreover, we require the following assumptions on the reaction term ${\bf F}:\mathbb{R}^n\longrightarrow \mathbb{R}^n$:
\begin{subequations}\label{assum_K}
\allowdisplaybreaks
 \begin{align}
 \begin{split}\label{K1}
 \bullet \,\,&{\bf F} \text{ is continuous with }\, {\bf F}(\mathbf{0}) = {\bf F}({\bf \bar{w}}) = {\bf 0},\,\text{ and there is no }\,{\bf \nu}\,\text{ such that }{\bf F}({\bf \nu}) = \mathbf{0}\,\\
 &\text{ and } \, \mathbf{0}< \nu <{\bf \bar{w}};
 \end{split}\\
  \begin{split}\label{K2}
\bullet \,\,& \text{ The system is cooperative, i.e., each component }\, f_i \,\text{ of }\, {\bf F}\,\text{ is nondecreasing in all the }\\
&\text{ components of }\, \mathbf{w}, \,\text{ with the possible exception of the } i\text{th one}; 
\end{split}\\
 \begin{split}\label{K3}
 \bullet \,\,& {\bf F}(\mathbf{w})\,\text{ is piecewise continuously differentiable in }\, \mathbf{w}\,\text{ for }\,\mathbf{0} \leq \mathbf{w}\leq {\bf \bar{w}}\,\text{ and differentiable at } \mathbf{0}\,,\\
&\text{ and the matrix }\, {\bf F}'(\mathbf{0})\,\text{ is irreducible with a positive eigenvalue;}  
 \end{split}\\
 \begin{split}\label{K4}
 \bullet \,\,&\text{There exists } \, a>0\,, \sigma >1\,, \text{and } {\bf r}>0\,\text{ such that }\\[0.1cm] 
 &\hspace{5.5cm}{\bf F}(\mathbf{w}) \geq {\bf F}'(\mathbf{0})\mathbf{w} - a\|\mathbf{w}\|^{\sigma}\mathbf{1}, \\
&\text{for all }\,\mathbf{0}\leq \mathbf{w} \leq {\bf r}, \text{ where } \mathbf{1} \text{ denotes the vector with all entries equal to 1;}
  \end{split}\\
  \begin{split}\label{K5}
\bullet \,\,&\text{For any } \,\eta >0, \\[0.1cm]
&\hspace{5cm}{\bf F}(\min\{\eta \mathbf{w})(\rho),{\bf 1}\} \leq \eta\, D{\bf F}(\mathbf{0}) \mathbf{w}(\rho),\,\\
& \text{for all} \, \rho \in (0,\rho^*]\,,\text{ where }\,\rho^* \, \text{ is the value of }\,\rho\, \text{ at which }\, \lambda_+(\rho)/\rho\,\text{ attains its infimum.}
 \end{split}
\end{align}
\end{subequations}

In all the assumptions, the operations are defined component-wise. We show in the following lemma that the general go-or-grow model~\eqref{general-model}, with the appropriate assumptions, is a special case of the reaction--diffusion equation~\eqref{super-general-model}.

\begin{lemma}
    If the general go-or-grow model~\eqref{general-model} satisfies \eqref{TW1}, \eqref{TW2}--\eqref{TW3}, and \eqref{TW5}--\eqref{TW4}, then \eqref{general-model} also satisfies the set of assumptions \eqref{assum_K}.
\end{lemma}
\begin{proof} Condition \eqref{K1}, and the first part of condition \eqref{K3} are directly satisfied by the assumptions \eqref{TW1} and \eqref{TW2}. In addition, during the wavefront calculation (Section~\ref{sec:travwave_unifyingframework}), we computed the eigenvalue of the matrix $A = \rho^2D + {\bf F}'(0,0)$ \eqref{eq:principal_eig_A}, which for $\rho = 0$ reads
\begin{align*}
    \lambda_A(0) =& \frac{g(0) - \mu -\alpha(0,0) - \beta(0,0)}{2} \\
    &+ \frac{\sqrt{\big[g(0) - \mu -\alpha(0,0) - \beta(0,0)\big]^2 + 4\big[(g(0) -\beta(0,0))(\mu +\alpha(0,0))+ \alpha(0,0)\beta(0,0)\big]}}{2}\\
    &>\frac{g(0) - \mu -\alpha(0,0) - \beta(0,0)}{2} + \frac{\sqrt{\big[g(0) - \mu -\alpha(0,0)-  \beta(0,0)\big]^2}}{2} \geq 0
\end{align*}
completing the requirements for \eqref{K3}.

For assumption \eqref{K2}, consider the Jacobian of ${\bf F}$ with components
\begin{equation*}
    \begin{split}
& {\bf F}'_{11}(u,v) =-(\mu -\alpha(u,v)) - \dfrac{\partial}{\partial u}(\alpha(u,v)) u + \dfrac{\partial}{\partial u} \beta(u,v) v , \\[0.2cm]
& {\bf F}'_{12}(u,v) =\beta(u,v) -\dfrac{\partial}{\partial v} \alpha(u,v)u + \dfrac{\partial}{\partial v}(\beta(u,v)) v , \\[0.2cm]
& {\bf F}'_{21}(u,v)= \alpha(u,v) + \dfrac{\partial}{\partial u} (\alpha(u,v))u + \dfrac{\partial}{\partial u}(g(u+v) -\beta(u,v))v , \\[0.2cm]
& {\bf F}'_{22}(u,v) = (g(u+v) -\beta(u,v)) +  \dfrac{\partial}{\partial v}(g(u+v) -\beta(u,v))v + \dfrac{\partial}{\partial v}\alpha(u,v)u .
    \end{split}
\end{equation*}
\normalsize
The cooperativity assumption is satisfied by having the off-diagonal terms greater than or equal to zero, which is exactly assumption \eqref{TW3}.

Now, the more technical assumptions \eqref{K4} and \eqref{K5} need to be satisfied. Considering the first inequality of assumption \eqref{TW4}, rearranging terms, and adding $-\mu\, u$ to both sides yields
\begin{align*}
    -(\mu +\alpha(u,v))u + \beta(u,v) \geq -(\mu +\alpha(0,0))u + \beta(0,0)v - a|u+v|^{\sigma},
\end{align*}
which is exactly ${\bf F}_1(u,v) \geq \big({\bf F}'_{11}(0,0),{\bf F}'_{12}(0,0)\big)\cdot\begin{pmatrix}
    u \\ v
\end{pmatrix} - a|u +v|^{\sigma}$. For the second component, considering the second inequality of assumption \eqref{TW4} and rearranging terms yields 
\begin{align*}
    g(u+v)v +\alpha(u,v) u -\beta(u,v)v \geq g(0) -\alpha(0,0)u + \beta(0,0)v- a|u+v|^{\sigma},
\end{align*}
which is exactly ${\bf F}_2(u,v) \geq \big({\bf F}'_{21}(0,0),{\bf F}'_{22}(0,0)\big)\cdot\begin{pmatrix}
    u \\ v
\end{pmatrix} - a|u +v|^{\sigma}$.
Finally for assumption \eqref{K5}, we compute
\begin{align*}
    {\bf F}'(0,0)\nu(\rho) &= 
    \begin{pmatrix}
        -(\mu - \alpha(0,0)) & \beta(0,0) \\[0.5cm]
        \alpha(0,0) & g(0) - \beta(0,0)
    \end{pmatrix}
    \begin{pmatrix}
        \dfrac{\lambda_A(\rho) -g(0) +\beta(0,0)}{\alpha(0,0)} \\[0.3cm] 1
    \end{pmatrix}\\[0.2cm] &= \begin{pmatrix}
        g(0) - \mu\nu_1(\rho) - \lambda_A(\rho) \\[0.5cm] \lambda_A(\rho)
    \end{pmatrix}.
\end{align*}
For the first component, taking the middle inequality of assumption~\eqref{TW5}, i.e.,
\begin{align} \label{TW5K}
    -\alpha\big(\eta\nu_1(\rho),\eta\big)\nu_1(\rho) + \beta\big(\eta\nu_1(\rho),\eta\big) \leq g(0) -\lambda_A(\rho) ,
\end{align}
adding $-\mu\nu_1(\rho)$ to both sides, and multiplying by $\eta > 0$, we get
\begin{align*}
 -\alpha\big(\eta\nu_1(\rho),\eta\big)\eta-\mu\nu_1(\rho)\eta + \beta\big(\eta\nu_1(\rho),\eta\big)\eta \leq \big(g(0) -\mu\nu_1(\rho)-\lambda_A(\rho)\big)\eta.
\end{align*}
Notice that this inequality is exactly ${\bf F}_1(\eta\nu_1(\rho),\eta) \leq \eta \big({\bf F}'_{11}(0,0),{\bf F}'_{12}(0,0)\big)\cdot(\eta\nu_1(\rho),\eta)^T$. For the second component, we consider the first inequality of assumption~\eqref{TW5}, i.e.,
\begin{align*}
    g\big(\eta + \eta \nu_1(\rho)\big) \leq -\alpha\big(\eta\nu_1(\rho),\eta\big)+ \beta\big(\eta\nu_1(\rho),\eta\big) + \lambda_A(\rho) . 
\end{align*}
Rearranging terms and multiplying by $\eta$, we get
\begin{align*}
    \alpha\big(\eta\nu_1(\rho),\eta\big)\eta- \beta\big(\eta\nu_1(\rho),\eta\big)\eta + g\big(\eta + \eta \nu_1(\rho)\big)\eta \leq \lambda_A(\rho)\eta,
\end{align*}
which is ${\bf F}_2(\eta\nu_1(\rho),\eta) \leq \eta \big({\bf F}'_{21}(0,0),{\bf F}'_{22}(0,0)\big)\cdot(\eta\nu_1(\rho),\eta)^T$.

Therefore, if \eqref{general-model} satisfies \eqref{TW1}, \eqref{TW2}--\eqref{TW3}, and \eqref{TW5}--\eqref{TW4}, then \eqref{general-model} also satisfies \eqref{assum_K}.
\end{proof}

Since the general go-org-grow model~\eqref{general-model} satisfies the general conditions \eqref{assum_K}, we can apply the following lemma to obtain information on the spreading speed.
\begin{lemma}[\cite{liang2007asymptotic,fang2009monotone,weinberger2002analysis}]\label{lemma_Liang}
Let $\bf F$ satisfy \eqref{K1}--\eqref{K3}, and consider the initial condition $\phi\in \mathcal{C}_{{\bf \bar{w}}}= \{(\phi_1,\dots,\phi_n) \in C(\mathbb{R};\mathbb{R}^n): 0\leq \phi_i(z) \leq \bar{w}_i\; \forall z\in \mathbb{R}\}$ and $ {\bf w}(t,x;\phi)$ be the unique solution of \eqref{super-general-model} through $\phi$. Then there exists a real number $0<\Tilde{c} \leq \bar{c}$, where
\begin{align}
    \bar{c} =\inf_{\mu > 0 }\Phi(\mu),
\end{align}
such that the following statements are valid:
\begin{itemize}
    \item (i)  $\Phi(\mu) = \dfrac{ \lambda_A(\mu) }{\mu}$ is decreasing as $\mu$ nears $0$ and tends to infinity as $\mu\rightarrow 0$; $\Phi'(\mu)$ changes sign at most once in $(0,\infty)$ and $\lim\limits_{\mu \rightarrow \infty} \Phi(\mu) = + \infty$;
    \item (ii) If $\phi$ has compact support, then $\lim\limits_{\substack{t\rightarrow \infty\\ |x|\geq ct}}  {\bf w}(t,x;\phi)  = 0,$ for all $c > \Tilde{c} $;
    \item (iii) For any $c\in(0,\Tilde{c})$ and  any vector $\mathbf{r}> \mathbf{0}$, there is a positive number $R_r$ such that for any $\phi \in \mathcal{C}_{{\bf \bar{w}}}$ with $\phi \leq \mathbf{r}$ on an interval of length $2R_{\mathbf{r}}$, it holds that $\lim\limits_{\substack{t\rightarrow \infty\\|x|\leq ct}} {\bf w}(t,x;\phi)  = {\bf\bar{w}}$;
    \item (iv) If, in addition, ${\bf F}(\min\{\rho \nu(\mu^*),\textbf{1}\} )\leq \rho {\bf F}'({\bf 0})\nu(\mu),$ for all $\rho >0,$ then $\Tilde{c} = \bar{c}$, where $\mu^*$ is the value of $\mu$ at which $\Phi(\mu)$ attains its infimum.
\end{itemize}
\end{lemma}

Finally, we use the following theorem to prove the existence of traveling wave solutions.
\begin{theorem}[Theorem 3.1 in \cite{fang2009monotone}]\label{Teo_Fang2009} Assume the set of assumptions \eqref{assum_K} hold, then for each $c\geq\bar{c}$, system \eqref{super-general-model} has a nondecreasing wavefront $W(x+ct)$ connecting $\textbf{0}$ and ${\bf \bar{w}}$; while for any $c\in (0,\bar{c})$, there is no wavefront $W(x+ct)$ connecting $\textbf{0}$ and ${\bf \bar{w}}$
\end{theorem}

Applying Lemma~\ref{lemma_Liang} and Theorem~\ref{Teo_Fang2009} to the general go-org-grow model \eqref{general-model}, with assumptions \eqref{TW1}, \eqref{TW2}--\eqref{TW3}, and \eqref{TW5}--\eqref{TW4}, proves Lemma~\ref{phiProp} and Theorem~\ref{general-model-cooperative-TW}. 

\end{appendices}

\printbibliography

\end{document}